\documentclass[11pt]{llncs}
\usepackage{verbatim}
\usepackage{cite}
\usepackage[usenames]{color}
\usepackage{amssymb,amsmath,enumerate}

\usepackage{tikz}
\usepackage{pgflibraryshapes}
\usetikzlibrary{snakes}
\tikzstyle{tre}=[circle,draw,minimum size=5.1mm]
\tikzstyle{hyb}=[rectangle,draw,minimum size=5.1mm]
\tikzstyle{indef}=[rectangle,rounded corners,draw,minimum size=5.1mm]
\newcommand{\etq}[1]{%
\draw (#1) node {$#1$};
}
\newcommand{\etqsm}[1]{%
\draw (#1) node {\small $#1$};
}

\newcommand{\etqfn}[1]{%
\draw (#1) node {\footnotesize $#1$};
}

\newcommand{\pathgr}{\!\rightsquigarrow\!{}}

\renewcommand{\leq}{\leqslant}
\renewcommand{\geq}{\geqslant}
\renewcommand{\le}{\leqslant}
\renewcommand{\ge}{\geqslant}
\newcommand{\NN}{\mathbb{N}}
\newcommand{\RR}{\mathbb{R}}
\newcommand{\TCTC}{\mathrm{TCTC}}
\begin{document}

\title{Path lengths in tree-child time consistent hybridization networks}

\author{Gabriel Cardona\inst{1} \and Merc\`e Llabr\'es\inst{1}\inst{2} \and Francesc Rossell\'o\inst{1}\inst{2} \and
Gabriel Valiente\inst{2}\inst{3}} \authorrunning{G. Cardona et al.}
\institute{Department of Mathematics and Computer Science, University
of the Balearic Islands, E-07122 Palma de Mallorca, Spain
\and
Research Institute of Health Science (IUNICS),  E-07122 Palma de Mallorca, Spain
\and
Algorithms, Bioinformatics, Complexity and Formal Methods Research
Group, Technical University of Catalonia, E-08034 Barcelona, Spain}

\maketitle

\begin{abstract}
Hybridization networks are representations of evolutionary
histories that allow for the inclusion of reticulate events like
recombinations, hybridizations, or lateral gene transfers.  The recent
growth in the number of hybridization network reconstruction
algorithms has led to an increasing interest in the definition of
metrics for their comparison  that can be used
to assess the accuracy or robustness of these methods.  In this paper
we establish some basic results that make it possible the
generalization to tree-child time consistent (TCTC) hybridization
networks of some of the oldest known metrics for phylogenetic trees:
those based on the comparison of the vectors of path lengths between
leaves.  More specifically, we associate to each hybridization network
a suitably defined vector of `splitted' path lengths between its
leaves, and we prove that if two TCTC hybridization networks have the
same such vectors, then they must be isomorphic.  Thus, comparing
these vectors by means of a metric for real-valued vectors defines a
metric for TCTC hybridization networks.  We also consider the case of
fully resolved hybridization networks, where we prove that simpler,
`non-splitted' vectors can be used.
\end{abstract}

\setcounter{footnote}{0}
\section{Introduction}
\label{sec:intro}

An evolutionary history is usually modelled by means of a rooted
phylogenetic tree, whose root represents a common ancestor of
all species under study (or whatever other taxonomic units are considered: genes,
proteins,\ldots), the leaves, the extant species, and the
internal nodes, the ancestral species.  But phylogenetic trees can
only cope with speciation events due to mutations, where each species
other than the universal common ancestor has only one parent in the
evolutionary history (its parent in the tree).  It is clearly
understood now that other speciation events, which cannot be properly
represented by means of single arcs in a tree, play an important role in
evolution~\cite{doolittle:99}.  These are \emph{reticulation events}
like genetic recombinations, hybridizations, or lateral gene
transfers, where a species is the result of the interaction between
two parent species.  This has lead to the introduction of
\emph{networks} as models of phylogenetic histories that capture these
reticulation events side by side with the classical mutations.

Contrary to what happens in the phylogenetic trees literature, where
the basic concepts are well established, there is still some
lack of consensus about terminology in the field of `phylogenetic networks'
\cite{huson.ea:06}.  Following \cite{semple:07}, in this paper we use
the term \emph{hybridization network} to denote the most general model
of reticulated evolutionary history: a directed acyclic graph with
only one root, which represents the last universal common ancestor and which we assume, thus,  of
out-degree greater than 1.  In such a graph, nodes represent species
(or any other taxonomy unit) and arcs represent direct descendance.  A
node with only one parent (a \emph{tree node}) represents a species
derived from its parent species through mutation, and a node with more
than one parent (a \emph{hybrid node}) represents a species derived
from its parent species through some reticulation event.

The interest in representing phylogenetic histories by means of
networks has lead to many hybridization network reconstruction
methods
\cite{gusfield.ea:fine.structure:2004,gusfield.ea:galled.trees:2004,huson:RECOMB07,moret.ea:2004,nakhleh.ea:2003,nakhleh.ea:2005,song.ea:2005,wang.ea:2001}.
These reconstruction methods often search for hybridization networks
satisfying some restriction, like for instance to have as few hybrid
nodes as possible (in \emph{perfect phylogenies}), or to have their
reticulation cycles satisfying some structural restrictions (in
\emph{galled trees} and \emph{networks}).  Two popular and
biologically meaningful such restrictions are the \emph{time
consistency} \cite{baroni.ea:sb06,moret.ea:2004}, the possibility of
assigning times to the nodes in such a way that tree children exist
later than their parents and hybrid children coexist with their
parents (and in particular, the parents of a hybrid species
coexist in time), and the \emph{tree child condition}
\cite{cardona.ea:07b,willson:07b}, that imposes that every non-extant
species has some descendant through mutation alone.  The
\emph{tree-child time consistent} (TCTC) hybridization networks
 have been recently proposed as the class
where meaningful phylogenetic networks should be
searched~\cite{willson:07}.  Recent simulations (reported in
\cite{valiente:08}) have shown that over  64\% of 4132
hybridization networks obtained using the coalescent model
\cite{coalescent} under various population and sample sizes, sequence
lengths, and recombination rates, were TCTC: the percentage
of TCTC networks among the time consistent networks obtained in these
simulations increases to 92.8\%.

The increase in the number of available hybridization networks
reconstruction algorithms has made it necessary the introduction of
methods for the comparison of hybridization networks to be used in
their assessment, for instance by comparing inferred networks with
either simulated or true phylogenetic histories, and by evaluating the
robustness of reconstruction algorithms when
adding new species~\cite{moret.ea:2004,woolley.ea:2008}.  This has
lead recently to the definition of several metrics defined on different
classes of hybridization networks
\cite{cardona.ea:sbTSTC:2008,comparison1,comparison2,cardona.ea:07a,cardona.ea:07b,moret.ea:2004,nakhleh.ea:03psb}.
All these metrics generalize in one way or another well-known metrics
for phylogenetic trees.

Some of the most popular metrics for phylogenetic trees are based on
the comparison of the vectors of path lengths between leaves
\cite{bluis.ea:2003,farris:sz69,farris:sz73,phipps:sz71,steelpenny:sb93,willcliff:taxon71}.
Introduced in the early seventies, with different names depending on
the author and the way these vectors are compared, they are globally
known as \emph{nodal distances}.  Actually, these vectors of paths
lengths only \emph{separate} (in the sense that equal vectors means isomorphic
trees), on the one hand, unrooted phylogenetic trees, and, on the
other hand, fully resolved rooted phylogenetic trees, and therefore, as far as rooted phylogenetic trees goes,
the distances defined through these vectors are only true metrics for  fully resolved  trees.
These metrics were recently generalized to arbitrary rooted
phylogenetic trees \cite{cardona.ea:08a}.  In this generalization,
each path length between two leaves was replaced by the pair of
distances from the leaves to their least common ancestor,
and the vector of paths lengths between leaves was replaced by the
\emph{splitted path lengths} matrix obtained in this way.  These
matrices separate arbitrary rooted phylogenetic trees, and therefore
the \emph{splitted nodal distances} defined through them are indeed
metrics on the space of rooted phylogenetic trees.

In a recent paper \cite{comparison2} we have generalized these
splitted nodal distances to TCTC hybridization networks with all their
hybrid nodes of out-degree 1.
The goal of this paper is to go one step beyond in two directions: to generalize to the TCTC hybridization
networks setting both the classical nodal distances for fully resolved
rooted phylogenetic trees and the new splitted nodal distances for
rooted phylogenetic trees.  Thus, on the one hand, we introduce a
suitable generalization of the vectors of path lengths between leaves
that separate \emph{fully resolved} (where every non extant species
has exactly two children, and every reticulation event involves
exactly two parent species) TCTC hybridization networks.  On the other
hand, we show that if we split these new path lengths in a suitable
way and we add a bit of extra information, the resulting vectors
separate arbitrary TCTC hybridization networks. Then, the vectors 
obtained in both cases can be used to define metrics that generalize, 
respectively, the nodal distances for fully resolved rooted 
phylogenetic trees and the splitted nodal distances for rooted 
phylogenetic trees. 

The key ingredient in the proofs of our main results is the use of
sets of suitable reductions that applied to  TCTC hybridization
networks with $n$ leaves and $m$ internal nodes produce TCTC
hybridization networks with either $n-1$ leaves or with $n$ leaves
and $m-1$ internal nodes (in the fully resolved case, the reductions we use are
specifically tailored to make them remove always one leaf).  Similar
sets of reductions have already been introduced for TCTC hybridization
networks with all their hybrid nodes of out-degree 1
\cite{comparison2} and for \emph{tree sibling} (where every hybrid
node has a tree sibling) time consistent hybridization networks with
all their hybrid nodes of in-degree 2 and out-degree 1 \cite{cardona.ea:sbTSTC:2008}, and they have 
been proved useful in those contexts not only to establish properties of 
the corresponding networks by algebraic induction, but also to 
generate in a recursive way all networks of the type under 
consideration. We hope that the reductions introduced in this paper 
will find similar applications elsewhere.

\section{Preliminaries}
\label{sec:prel}

\subsection{Notations on DAGs}
\label{subsec:dags}

Let $N=(V,E)$ denote in this subsection a directed acyclic (non-empty, finite) graph; a
\emph{DAG}, for short.  A node $v\in V$ is a \emph{child} of $u\in V$ if $(u,v)\in E$;
we also say in this case that $u$ is a \emph{parent} of $v$.  All children of the
same parent are said to be \emph{sibling} of each other.  

Given a node $v\in V$, its \emph{in-degree} $\deg_{in}(v)$ and its
\emph{out-degree} $\deg_{out}(v)$ are, respectively, the number of its
parents and the number of its children.  The \emph{type} of $v$ is the
ordered pair $(\deg_{in}(v),\deg_{out}(v))$.  A node $v$ is a \emph{root} when $\deg_{in}(v)=0$, a \emph{tree
node} when $\deg_{in}(v)\leq 1$, a \emph{hybrid node} when
$\deg_{in}(v)\geq 2$, a  \emph{leaf} when $\deg_{out}(v)=0$, 
\emph{internal} when
$\deg_{out}(v)\geq 1$, and  \emph{elementary} when $\deg_{in}(v)\leq 1$ and $\deg_{out}(v)= 1$. A \emph{tree arc} (respectively, a
\emph{hybridization arc}) is an arc with head a tree node
(respectively, a hybrid node). 
 A DAG $N$ is
\emph{rooted} when it has only one {root}.


%

A \emph{path} on $N$ is a sequence of nodes $(v_0,v_1,\dots,v_k)$ such
that $(v_{i-1},v_{i})\in E$ for all $i=1,\dots,k$.  We call $v_0$ the \emph{origin} of the
path, $v_{1},\ldots,v_{k-1}$ its \emph{intermediate nodes}, and
$v_{k}$ its \emph{end}.  The \emph{length} of the path
$(v_0,v_1,\dots,v_k)$ is $k$, and it is \emph{non-trivial} if $k\ge
1$.  The \emph{acyclicity} of $N$ means that it does not contain \emph{cycles}: non-trivial paths from a node to itself.

 We denote by $u\pathgr v$ any path
with origin $u$ and end $v$. Whenever
there exists a path $u\pathgr v$, we shall say that $v$ is a \emph{descendant} of $u$ and
also that $u$ is an \emph{ancestor} of $v$. When the path $u\pathgr v$ is non-trivial, we say that $v$ is a  \emph{proper} descendant of $u$ and
that $u$ is an \emph{proper}  {ancestor} of $v$.
The \emph{distance} from a node $u$ to a descendant $v$ is the length of a shortest
path from $u$ to $v$. 

The \emph{height} $h(v)$ of a node $v$ in a DAG $N$ is the largest
length of a path from $v$ to a leaf.  The absence of cycles implies
that the nodes of a DAG can be stratified by means of their heights:
the nodes of height 0 are the leaves, the nodes of height 1 are those
nodes all whose children are leaves, the nodes of height 2 are those
nodes all whose children are leaves and nodes of height 1, and so on.
If a node has height $m>0$, then all its children have height smaller
than $m$, and at least one of them has height exactly $m-1$.

A node $v$ of $N$ is a \emph{strict descendant} of a node $u$ if it is
a descendant of it, and every path from a root of $N$ to $v$ contains
the node $u$: in particular, we understand every node as a strict
descendant of itself.  When $v$ is a strict descendant of $u$, we
also say that $u$ is a \emph{strict ancestor} of $v$.

The following lemma will be used several times in this paper.

\begin{lemma}
\label{lem:str-int}
Let $u$ be a proper strict ancestor of a node $v$ in a DAG $N$, and let $w$ be an intermediate node in a path $u\pathgr v$. Then, $u$ is also a strict ancestor of $w$.
\end{lemma}

\begin{proof}
Let $r\pathgr w$ be a path from a root of $N$ to $w$, and concatenate to it the piece $w\pathgr v$ of the path $u\pathgr v$ under consideration. This yields a path $r\pathgr v$ that must contain $u$. Since $u$ does not appear in the piece $w\pathgr v$, we conclude that it is contained in the path $r\pathgr w$.
This proves that every path  from a root of $N$ to $w$ contains
the node $u$.
\end{proof}

For
every pair of nodes $u,v$ of $N$:
\begin{itemize}
\item $CSA(u,v)$ is the set of all common ancestors of $u$ and $v$
that are strict ancestors of at least one of them;

\item the \emph{least common semi-strict ancestor} (\emph{LCSA}) of
$u$ and $v$, in symbols $[u,v]$, is the node in $CSA(u,v)$ of minimum height.
\end{itemize}
The LCSA of two nodes $u,v$ in a phylogenetic network is well defined and it is unique: it is actually
the unique element of $CSA(u,v)$ that is a descendant of all elements of this set
\cite{ comparison1}.  The following result on LCSAs will be used often. It is the generalization to DAGs of Lemma 6 in
\cite{comparison2}, and we include its easy proof for the sake of completeness.

\begin{lemma}
\label{lem:LCSA}
Let $N$ be a DAG and let $u,v$ be a pair of nodes of
$N$ such that $v$ is not a descendant of $u$.  If $u$ is a tree node
with parent $u'$, then $[u,v]=[u',v]$.
\end{lemma}

\begin{proof}
We shall prove that $CSA(u,v)=CSA(u',v)$.

Let $x\in CSA(u,v)$.
Since $u$ is not an ancestor of $v$, $x\neq u$ and hence any path $x\pathgr u$ is
non-trivial.  Then, since $u'$ is the only parent of $u$, it appears
in this path, and therefore $x$ is also an ancestor of $u'$.  This
shows that $x$ is a common ancestor of $u'$ and $v$.
Now, if $x$ is a strict ancestor of $v$, we already conclude that $x\in CSA(u',v)$, while
 if $x$ is a strict ancestor of $u$,  it will be also a strict ancestor of $u'$ 
 by Lemma \ref{lem:str-int},
and hence $x\in CSA(u',v)$, too.  This proves that $CSA(u,v)\subseteq CSA(u',v)$

Conversely, let $x\in CSA(u',v)$.
Since $u'$ is the parent of $u$, it is clear that $x$ is a common ancestor of $u$
and $v$, too.  If $x$ is a strict ancestor of $v$, this implies that $x\in CSA(u,v)$. If
 $x$ is a strict ancestor of $u'$, then it is
also a strict ancestor of $u$ (every path
$r\pathgr u$ must contain the only parent $u'$ of $u$, and then $x$ will belong to the piece $r\pathgr u'$ of
the path $r\pathgr u$), and therefore $x\in CSA(u,v)$, too. This finishes the proof of the equality.
\end{proof}

Let $S$ be any non-empty finite set of \emph{labels}.  We say that the DAG $N$
is \emph{labeled in $S$}, or that it is an \emph{$S$-DAG}, for short,
when its leaves are bijectively labeled by elements of $S$.  
Although in real applications  the set $S$
would correspond to a given set of extant  taxa, for the sake of simplicity we  shall assume henceforth that $S=\{1,\ldots,n\}$, with $n=|S|$.
We shall always identify, usually without any further notice,
each leaf of an $S$-DAG with its label in $S$.

Two $S$-DAGs $N,N'$ are \emph{isomorphic}, in symbols $N\cong N'$, when they are isomorphic as directed graphs and the isomorphism maps each leaf in $N$ to the leaf with the same label in $N'$.

\subsection{Path lengths in phylogenetic trees}
\label{sec:trees}

A \emph{phylogenetic tree} on a set $S$ of taxa is a rooted $S$-DAG
without hybrid nodes and such that its root is non-elementary.  A phylogenetic tree is \emph{fully resolved}, or
\emph{binary}, when every internal node has out-degree 2.  Since all ancestors of a node in a phylogenetic tree are strict, the LCSA
$[u,v]$ of two nodes $u,v$ in a phylogenetic tree is simply their
\emph{least common ancestor}: the unique common ancestor of them
that is a descendant of every other common ancestor of them.

Let $T$ be a phylogenetic tree on the set $S=\{1,\ldots,n\}$.  For
every $i,j\in S$, we shall denote  by
$\ell_T(i,j)$ and $\ell_T(j,i)$ the lengths of the paths
$[i,j]\pathgr i$ and $[i,j]\pathgr j$, respectively. In particular, 
$\ell_T(i,i)=0$  for every $i=1,\ldots,n$.

\begin{definition}
Let $T$ be a phylogenetic tree on the set $S=\{1,\ldots,n\}$.  The
\emph{path length} between two leaves $i$ and $j$ is
$$
L_T(i,j)=\ell_T(i,j)+\ell_T(j,i).
$$
The \emph{path lengths  vector} of $T$ is  the  vector 
$$
L(T)=\big(L_T(i,j)\big)_{1\leq i<j\leq n}\in \NN^{n(n-1)/2}
$$
with its entries ordered lexicographically in $(i,j)$.
\end{definition}

The following result is a special case of Prop.~2 in
\cite{cardona.ea:08a}.
\begin{proposition}
\label{prop:nod-bintree}
Two fully resolved phylogenetic trees on the same set $S$ of taxa are
isomorphic if, and only if, they have the same path lengths vectors.
\qed
\end{proposition}

The thesis in the last result is false for arbitrary
phylogenetic trees.  Consider for instance the phylogenetic trees with
Newick strings
\texttt{(1,2,(3,4));} and
\texttt{((1,2),3,4);} depicted\footnote{Henceforth, in graphical representations of DAGs, hybrid nodes are represented by
squares, tree nodes by circles, and indeterminate nodes, that is,
nodes that can be of tree or hybrid type, by squares with rounded corners.}
 in Fig.~\ref{fig:contr1}.  It is straightforward to check that they
have the same path lengths vectors, but they are not isomorphic.

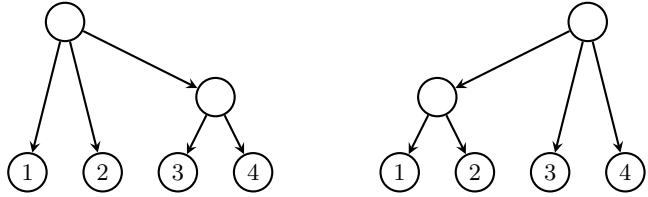
\begin{figure}[htb]
\begin{center}
            \begin{tikzpicture}[thick,>=stealth,scale=0.5]
              \draw(0,0) node[tre] (1) {}; \etq 1
              \draw(2,0) node[tre] (2) {}; \etq 2
              \draw(4,0) node[tre] (3) {}; \etq 3
              \draw(6,0) node[tre] (4) {}; \etq 4
              \draw(5,2) node[tre] (a) {};   
\draw(1,4) node[tre] (r) {};        
            \draw[->] (r)--(1);
            \draw[->] (r)--(2);
                        \draw[->] (r)--(a);
            \draw[->] (a)--(3);
            \draw[->] (a)--(4);
            \end{tikzpicture}
  \qquad\qquad
            \begin{tikzpicture}[thick,>=stealth,scale=0.5]
              \draw(0,0) node[tre] (1) {}; \etq 1
              \draw(2,0) node[tre] (2) {}; \etq 2
              \draw(4,0) node[tre] (3) {}; \etq 3
              \draw(6,0) node[tre] (4) {}; \etq 4
              \draw(5,4) node[tre] (a) {};    
\draw(1,2) node[tre] (r) {};       
            \draw[->] (r)--(1);
            \draw[->] (r)--(2);
                        \draw[->] (a)--(r);
            \draw[->] (a)--(3);
            \draw[->] (a)--(4);
            \end{tikzpicture}
\end{center}
\caption{\label{fig:contr1} 
Two non-isomorphic phylogenetic trees with the same path
lengths vectors.}
\end{figure}

This problem was overcome in \cite{cardona.ea:08a} by replacing the path lengths vectors by the following matrices of distances.

\begin{definition}
\label{def:split-tree}
The \emph{splitted path lengths matrix} of $T$ is the $n\times n$ square
matrix  
$$
\ell(T)=\big(\ell_T(i,j)\big)_{i=1,\ldots,n\atop j=1,\ldots,n}\in \mathcal{M}_n(\NN).
$$
\end{definition}

Now, the following result is (again, a special case of) Theorem 11 in
\cite{cardona.ea:08a}.

\begin{proposition}
\label{prop:nod-arbtree}
Two phylogenetic trees on the same set $S$ of taxa are isomorphic if,
and only if, they have the same splitted path lengths matrices.  \qed
\end{proposition}

\section{TCTC networks}
\label{sec:tctc}

While the basic notion of phylogenetic tree is well established,
the notion of phylogenetic network is much less well defined \cite{huson.ea:06}.
The networks we consider in this paper are the (almost) most general possible ones:
 rooted $S$-DAGs with non-elementary root. Following  \cite{semple:07}, we shall call them
 \emph{hybridization networks}.  In these
hybridization networks, every node represents a different species, and the
arcs represent direct descendance, be it through mutation (tree arcs)
or through some reticulation event (hybridization arcs).

It is usual to forbid elementary nodes in hybridization networks  \cite{semple:07},
mainly because they cannot be reconstructed.  We allow them here for
two reasons.  On the one hand, because allowing them simplifies
considerably some proofs, as it will be hopefully clear in Section
\ref{sec:split}.  On the other hand, because, as Moret \textsl{et al}
point out \cite[\S 4.3]{moret.ea:2004}, they can be useful both from
the biological point of view, to include auto-polyploidy in the model,
as well as from the formal point of view, to make a phylogeny
satisfy other constraints, like for instance time consistency (see
below) or the impossibility of successive hybridizations. Of course, our main results apply without any modification to hybridization networks without elementary nodes as well.

Following \cite{comparison1},  by a \emph{phylogenetic network} on a set $S$ of taxa we understand
a  rooted $S$-DAG $N$ with non-elementary root
 where every hybrid node
has exactly one child, and it is a tree node. 
Although, from the mathematical point of view, phylogenetic networks are a special case of hybridization networks, from the point of view of
modelling they represent in a different way evolutive histories with
reticulation events:  in a phylogenetic
network, every tree node represents a different species and every hybrid
node, a reticulation event that gives rise to the species represented
by its only child.

A hybridization network $N=(V,E)$ is \emph{time consistent} when it
allows a \emph{temporal representation} \cite{baroni.ea:sb06}: a
mapping
$$
\tau:V\to \NN
$$
such that $\tau(u)<\tau(v)$ for every tree arc $(u,v)$ and
$\tau(u)=\tau(v)$ for every hybridization arc $(u,v)$.  Such a temporal
representation can be understood as an assignment of times to nodes
that strictly increases from parents to tree children and so that the
parents of each hybrid node coexist in time.

\begin{remark}
\label{rem:tico}
Let $N=(V,E)$ be a time consistent hybridization network, and let
$N_1=(V_1,E_1)$ be a hybridization network obtained by removing from
$N$ some nodes and all their descendants (as well as all arcs pointing
to any removed node).  Then $N_1$ is still time consistent, because
the restriction of any temporal representation $\tau:V\to \NN$ of $N$
to $V_1$ yields a temporal representation of $N_1$.
\end{remark}

A hybridization network satisfies the \emph{tree-child condition}, or
it is \emph{tree-child}, when every internal node has at least one
child that is a tree node (a \emph{tree child}).  So, tree-child
hybridization networks can be understood as general models of
reticulate evolution where every species other that the extant ones,
represented by the leaves, has some descendant through mutation.
Tree-child hybridization networks include galled
trees~\cite{gusfield.ea:fine.structure:2004,gusfield.ea:galled.trees:2004}
as a particular case \cite{cardona.ea:07b}.

A \emph{tree path} in a tree-child hybridization network is a non-trivial path such that
its end and all its intermediate nodes are tree nodes.  A node $v$ is
a \emph{tree descendant} of a node $u$ when there exists a tree path
from $u$ to $v$.  By~\cite[Lem.~2]{cardona.ea:07a}, every internal
node $u$ of a tree-child hybridization network has some tree descendant
leaf, and by~\cite[Cor.  4]{cardona.ea:07a} every tree descendant $v$
of $u$ is a strict descendant of $u$ and the path $u\pathgr v$ is
unique.

To simplify the notations, we shall call \emph{TCTC-networks} the
tree-child time consistent hybridization networks: these include the
tree-child time consistent phylogenetic networks, which were the
objects dubbed TCTC-networks in \cite{comparison1,comparison2}.  
Every phylogenetic tree is also a TCTC-network.
Let $\TCTC_{n}$ denote the class of all TCTC-networks on
$S=\{{1},\ldots,{n}\}$.

We prove now some basic properties of TCTC-networks that will be used later.

\begin{lemma}
\label{lem:hyb-nostr}
Let $u$ be a node of a TCTC-network $N$, and let $v$ be a child of $u$.
The node $v$ is a tree node if, and only if, it is a strict descendant of $u$.
\end{lemma}

\begin{proof}
Assume first that $v$ is a tree child of $u$. Since $u$ is the only parent of $v$, every non-trivial path ending in $v$ must contain $u$. This shows that $u$ is a strict ancestor of $v$.

Assume now that $v$ is a hybrid child of $u$ that is also a strict descendant of it, and let us see that this leads to a contradiction. Indeed, in this case the set $H(u)$ of hybrid children of $u$ that are strict descendants of it is non-empty, and we can choose a node $v_0$ in it of largest height.
Let $v_1$ be any parent of $v_0$ other than $u$. Since $u$ is a strict ancestor of $v_0$, it must be an ancestor of $v_1$, and since $u$ and $v_1$ have the hybrid child $v_0$ is common, they must have the same temporal representation, and therefore $v_1$ as well as all intermediate nodes in any path $u\pathgr v_1$ must be hybrid. Moreover, since $u$ is a strict ancestor of $v_0$, it is also a strict ancestor of $v_1$ as well as of any intermediate node in any path $u\pathgr v_1$ (by Lemma \ref{lem:str-int}). In particular, the child of $u$ in a path $u\pathgr v_1$ will belong to $H(u)$ and its height will be larger than the height of $v_0$, which is impossible.
\end{proof}

\begin{corollary}
\label{lem:root}
All children of the root of a TCTC-network are tree nodes.
\end{corollary}

\begin{proof}
Every node in a hybridization network is a strict descendant of the root. Then, Lemma \ref{lem:hyb-nostr} applies.
\end{proof}

The following result is the key ingredient in the proofs of our main
results; it generalizes to hybridization networks Lemma 3 in
\cite{comparison2}, which referred to phylogenetic networks.  A
similar result was proved in \cite{cardona.ea:sbTSTC:2008} for
 \emph{tree-sibling} (that is, where every hybrid node has a sibling that is a tree node) time consistent phylogenetic networks with all its hybrid
nodes of in-degree 2.

\begin{lemma}
\label{lem:key}
Every TCTC-network with more than one leaf contains at least one node $v$ satisfying one of
the following properties:
\begin{enumerate}[(a)]
\item $v$ is an internal tree node and all its children are tree leaves.

\item $v$ is a hybrid internal node, all its children are tree leaves,
and all its siblings are leaves or hybrid nodes.

\item $v$ is a hybrid leaf, and all its siblings are leaves or hybrid
nodes.
\end{enumerate}
\end{lemma}

\begin{proof}
Let $N$ be a TCTC-network and $\tau$ a temporal representation of it.
Let $v_0$ be an internal node of highest $\tau$-value and, among such
nodes, of smallest height.  The tree children of $v_0$ have strictly
higher $\tau$-value than $v$, and therefore they are leaves.  And the
hybrid children of $v_0$ have the same $\tau$-value than $v_0$ but
smaller height, and therefore they are also leaves.

Now:
\begin{itemize}
\item If $v_0$ is a tree node all whose children are tree nodes,
taking $v=v_0$ we are in case (a).

\item If $v_0$ is a hybrid node all whose children are tree nodes,
then its parents have its same $\tau$-value, which, we recall, is the
highest one.  This implies that their children ($v_0$'s siblings)
cannot be internal tree nodes, and hence they are leaves or hybrid
nodes.  So, taking $v=v_0$, we are in case (b).

\item If $v_0$ has some hybrid child, take as the node $v$ in the
statement this hybrid child: it is a leaf, and all its parents have
the same $\tau$-value as $v_0$, which implies, arguing as in the
previous case, that all siblings of $v$ are leaves or hybrid nodes.
Thus, $v$ satisfies (c).
\end{itemize}
\end{proof}

We introduce now some reductions for TCTC-networks.  Each of these
reductions applied to a TCTC-network with $n$ leaves and $m$ internal
nodes produces a TCTC-network with either $n-1$ leaves  and $m$ internal nodes or  with $n$ leaves and $m-1$
internal nodes, and given any TCTC-network with more than two leaves,
it will always be possible to apply to it some of these reductions.
This lies at the basis of the proofs by algebraic induction of the main results in
this paper.

Let $N$ be a TCTC-network with $n\geq 3$ leaves.
\begin{enumerate}
\item[(\textbf{U})]
Let $i$ be one tree leaf of $N$ and assume that its parent has only
this child.  The \emph{$U(i)$ reduction} of $N$ is the network
$N_{U(i)}$ obtained by removing the leaf $i$, together with its
incoming arc, and labeling with $i$ its former parent;
cf.~Fig.~\ref{fig:U-red}.  This reduction removes the only child of a
node, and thus it is clear that $N_{U(i)}$ is still a TCTC-network,
with the same number of leaves but one internal node less than $N$.

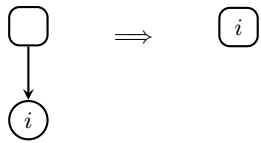
\begin{figure}[htb]
\centering
  \begin{tikzpicture}[thick,>=stealth,scale=0.5]
    \draw (0,0) node[indef] (u) {}; 
    \draw (0,-2.5) node[tre] (i) {}; \etqsm {i}
    \draw [->](u)--(i);
  \end{tikzpicture}
  \qquad
               \begin{tikzpicture}[thick,>=stealth,scale=0.5]
              \draw(0,-2.5) node{\ };
              \draw(0,0) node  {$\Longrightarrow$};
 \end{tikzpicture}
  \qquad
  \begin{tikzpicture}[thick,>=stealth,scale=0.5]
    \draw (0,0.3) node[indef] (i) {};  \etqsm {i}
              \draw(0,-2.5) node{\ };
    \end{tikzpicture}

\caption{\label{fig:U-red} 
The $U(i)$-reduction.}
\end{figure}

\item[(\textbf{T})]
Let $i,j$ be two sibling tree leaves of $N$ (that may, or may not,
have other siblings).  The \emph{$T(i;j)$ reduction} of $N$ is the
network $N_{T(i;j)}$ obtained by removing the leaf $i$, together with
its incoming arc; cf.~Fig.~\ref{fig:T-red}.  This reduction procedure
removes one tree leaf, but its parent $u$ keeps at least another tree
child, and if $u$ was the root of $N$ then it would not become
elementary after the reduction, because $n\geq 3$ and therefore, since
$j$ is a leaf, $u$ should have at least another child.  Therefore,
$N_{T(i;j)}$ is a TCTC-network with the same number of internal nodes as $N$ and $n-1$ leaves.

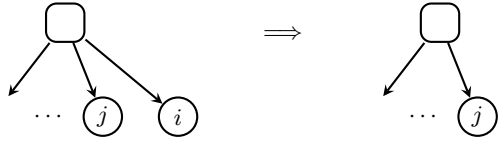
\begin{figure}[htb]
\centering
  \begin{tikzpicture}[thick,>=stealth,scale=0.5]
    \draw (0,0.5) node[indef] (u) {}; 
    \draw (1,-2) node[tre] (j) {}; \etqsm {j}
    \draw (3,-2) node[tre] (i) {}; \etqsm {i}
    \draw (-0.4,-2) node {$\cdots$};
    \draw (-1.7,-1.7) node  (x) {}; 
    \draw [->](u)--(i);
    \draw [->](u)--(j);
    \draw [->](u)--(x);
  \end{tikzpicture}
  \qquad
               \begin{tikzpicture}[thick,>=stealth,scale=0.5]
              \draw(0,-2) node{\ };
              \draw(0,0.5) node  {$\Longrightarrow $};
 \end{tikzpicture}
  \qquad
  \begin{tikzpicture}[thick,>=stealth,scale=0.5]
    \draw (0,0.5) node[indef] (u) {}; 
    \draw (1,-2) node[tre] (j) {}; \etqsm {j}
    \draw (-0.4,-2) node {$\cdots$};
    \draw (-1.7,-1.7)  node  (x) {}; 
    \draw [->](u)--(j);
    \draw [->](u)--(x);
  \end{tikzpicture}

\caption{\label{fig:T-red} 
The $T(i;j)$-reduction.}
\end{figure}

\item[(\textbf{H})] Let $i$ be a hybrid leaf of $N$, let
$v_1,\ldots,v_k$, with $k\geq 2$, be its parents, and assume that each one of these
parents has (at least) one tree leaf child: for every $l=1,\ldots,k$,
let $j_l$ be a tree leaf child of $v_l$.  The
\emph{$H(i;j_1,\ldots,j_k)$ reduction} of $N$ is the network
$N_{H(i;j_1,\dots,j_k)}$ obtained by removing the hybrid leaf $i$ and
its incoming arcs; cf.~Fig.~\ref{fig:H-red}.  This reduction procedure
preserves the time consistency and the tree-child condition (it
removes a hybrid leaf), and the root does not become elementary:
indeed, the only possibility for the root to become elementary is to
be one of the parents of $i$, which is impossible by Corollary \ref{lem:root}.  Therefore, $N_{H(i;j_1,\dots,j_k)}$ is a
TCTC-network with the same number of internal nodes as $N$ and $n-1$ leaves.

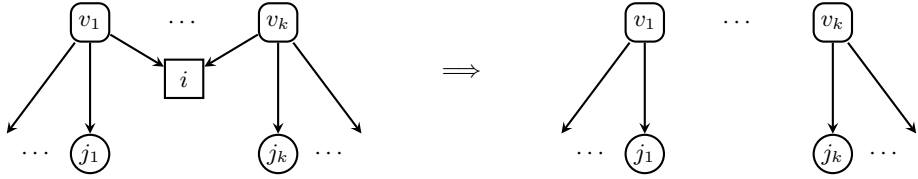
\begin{figure}[htb]
\centering
  \begin{tikzpicture}[thick,>=stealth,scale=0.5]
    \draw (0,0) node[hyb] (i) {}; \etqsm{i}
    \draw (-2.5,1.5) node[indef] (v_1) {}; \etqfn{v_1}
    \draw (-2.5,-2) node[tre] (j_1) {}; \etqfn{j_1}
      \draw (-4.9,-1.7) node  (x) {}; 
    \draw (-3.9,-2) node {$\cdots$};
  \draw (2.5,1.5) node[indef] (v_k) {}; \etqfn{v_k}
    \draw (2.5,-2) node[tre] (j_k) {}; \etqfn{j_k}
      \draw (4.9,-1.7) node  (y) {}; 
    \draw (3.9,-2) node {$\cdots$};
    \draw (0,1.5) node {$\cdots$};
    \draw[->](v_1) to (j_1);
    \draw[->](v_1) to (i);
    \draw[->](v_1) to (x);
    \draw[->](v_k) to (j_k);
    \draw[->](v_k) to  (i);    
    \draw[->](v_k) to  (y);    
    \end{tikzpicture}
  \qquad             \begin{tikzpicture}[thick,>=stealth,scale=0.5]
              \draw(0,1) node{\ };
              \draw(0,0.5) node  {$\Longrightarrow $};
              \draw(0,-2) node{\ };
 \end{tikzpicture}
 \qquad
  \begin{tikzpicture}[thick,>=stealth,scale=0.5]
   \draw (-2.5,1.5) node[indef] (v_1) {}; \etqfn{v_1}
    \draw (-2.5,-2) node[tre] (j_1) {}; \etqfn{j_1}
      \draw (-4.9,-1.7) node  (x) {}; 
    \draw (-3.9,-2) node {$\cdots$};
  \draw (2.5,1.5) node[indef] (v_k) {}; \etqfn{v_k}
    \draw (2.5,-2) node[tre] (j_k) {}; \etqfn{j_k}
      \draw (4.9,-1.7) node  (y) {}; 
    \draw (3.9,-2) node {$\cdots$};
    \draw (0,1.5) node {$\cdots$};
    \draw[->](v_1) to (j_1);
    \draw[->](v_1) to (x);
    \draw[->](v_k) to (j_k);
    \draw[->](v_k) to  (y);    
  \end{tikzpicture}

\caption{\label{fig:H-red} 
The $H(i;j_1,\ldots,j_k)$-reduction.}
\end{figure}

\end{enumerate}

We shall call the inverses of the U, T, and H reduction procedures,
respectively, the $\textrm{U}^{-1}$, $\textrm{T}^{-1}$, and
$\textrm{H}^{-1}$ \emph{expansions}, and we shall denote them by
$U^{-1}(i)$, $R^{-1}(i;j)$, and $H^{-1}(i;j_1,\ldots,j_k)$.  More
specifically, for every TCTC-network $N$:
\begin{itemize}
\item If $N$ has some leaf labeled $i$, the expansion $U^{-1}(i)$ can
be applied to $N$ and the resulting network $N_{U^{-1}(i)}$ is
obtained by unlabeling the leaf $i$ and adding to it a tree leaf child
labeled with $i$.  $N_{U^{-1}(i)}$ is always a TCTC-network.

\item If $N$ has no leaf labeled with $i$ and some tree leaf labeled
with $j$, the expansion $T^{-1}(i;j)$ can be applied to $N$, and the
resulting network $N_{T^{-1}(i;j)}$ is obtained by adding to the
parent of the leaf $j$ an new tree leaf child labeled with $i$.
$N_{T^{-1}(i;j)}$ is always a TCTC-network.

\item If $N$ has no leaf labeled with $i$ and some tree leaves labeled
with $j_1,\ldots,j_k$, $k\geq 2$, that are not sibling of each other,
the expansion $H^{-1}(i;j_1,\ldots,j_k)$ can be applied to $N$ and the
resulting network $N_{H^{-1}(i;j_1,\ldots,j_k)}$ is obtained by adding
a new hybrid node labeled with $i$ and arcs from the parents of
$j_1,\ldots,j_k$ to $i$.  $N_{H^{-1}(i;j_1,\ldots,j_k)}$ is always a
tree child hybridization network, but it need not be time consistent,
as the parents of $j_1,\ldots,j_k$ may  have different temporal
representations in $N$ (for instance, one of them could be a tree
descendant of another one).
\end{itemize}

The following result is easily deduced from the explicit descriptions
of the reduction and expansion procedures, and the fact that isomorphisms
preserve labels and parents.

\begin{lemma}
\label{lem:iso-red}
Let $N$ and $N'$ be two TCTC-networks.  If $N\cong N'$, then the
result of applying to both $N$ and $N'$ the same U reduction
(respectively, T reduction, H reduction, $\textrm{U}^{-1}$ expansion,
$\textrm{T}^{-1}$ expansion, or $\textrm{H}^{-1}$ expansion) are again
two isomorphic hybridization networks.

Moreover, if we apply an U reduction (respectively, T reduction, or H
reduction) to a TCTC-network $N$, and then we apply to the resulting
TCTC-network the inverse $\textrm{U}^{-1}$ expansion (respectively,
$\textrm{T}^{-1}$ expansion, or $\textrm{H}^{-1}$ expansion), we
obtain a TCTC-network isomorphic to $N$.  \qed
\end{lemma}

As we said above, every TCTC-network with at least 3 leaves allows the
application of some reduction.

\begin{proposition}\label{thm-reduction-possible}
Let $N$ be a TCTC-network with more than two leaves.  Then, at least
one U, R, or H reduction can be applied to $N$.
\end{proposition}

\begin{proof}
By Lemma~\ref{lem:key}, $N$ contains either an internal (tree or
hybrid) node $v$ all whose children are tree leaves, or a hybrid leaf
$i$ all whose siblings are leaves or hybrid nodes.
In the first case, we can apply to $N$ either the reduction $U(i)$ (if
$v$ has only one child, and it is the tree leaf $i$) or $T(i;j)$ (if
$v$ has at least two tree leaf children, $i$ and $j$).
In the second case, let $v_1,\ldots,v_k$, with $k\geq 2$, be the parents of $i$.  By
the tree child condition, each $v_l$, with $l=1,\ldots,k$, has some
tree child, and by the assumption on $i$, it will be a leaf, say $j_l$.  Then, we
can apply to $N$ the reduction $H(i;j_1,\ldots,j_k)$.
\end{proof}

Therefore, every TCTC-network with $n\geq 3$ leaves and $m$ internal nodes is obtained by the application of an $\textrm{U}^{-1}$,
$\textrm{T}^{-1}$, or $\textrm{H}^{-1}$ expansion
to a TCTC-network with either $n-1$ leaves or $n$ leaves and $m-1$ internal nodes.
This allows the recursive construction of all TCTC-networks from  TCTC-networks (actually,  phylogenetic trees) with 2 leaves and 1 internal node.

\begin{example}
Fig.~\ref{fig:seqred} shows how a sequence of reductions transforms a certain
TCTC-network $N$ with 4 leaves into a phylogenetic tree with  2 leaves.
The sequence of inverse expansions would then generate $N$ from this phylogenetic tree.  This sequence of expansions generating $N$ is, of course, not unique.
\end{example}
\begin{figure}[htb]
\begin{center}
\begin{tikzpicture}[thick,>=stealth,scale=0.45]
                \draw(0,0) node[tre] (1) {}; \etqsm 1
                   \draw(2,0) node[tre] (2) {}; \etqsm 2
                \draw(4,0) node[tre] (3) {}; \etqsm 3
                \draw(6,0) node[tre] (4) {}; \etqsm 4
  \draw(4,3) node[hyb] (A) {};  
\draw(1,3) node[tre] (a) {};
\draw(6,3) node[tre] (b) {};
\draw(3,6) node[tre] (r) {};
 \draw[->] (a)--(1);
  \draw[->] (a)--(A);
  \draw[->] (a)--(2);
   \draw[->] (b)--(A);
 \draw[->] (b)--(4);
 \draw[->] (A)--(3);
 \draw[->] (r)--(a);
 \draw[->] (r)--(b);
            \end{tikzpicture}
                  \begin{tikzpicture}[thick,>=stealth,scale=0.45]
              \draw(0,0) node{\ };
              \draw(0,3) node  {$\Longrightarrow$};
              \draw(0,4) node  {\footnotesize $U(3)$};
              \draw(0,6) node{\ };
 \end{tikzpicture}
 \begin{tikzpicture}[thick,>=stealth,scale=0.45]
                \draw(0,0) node[tre] (1) {}; \etqsm 1
                   \draw(2,0) node[tre] (2) {}; \etqsm 2
                \draw(6,0) node[tre] (4) {}; \etqsm 4
  \draw(4,3) node[hyb] (3) {};  \etqsm 3
\draw(1,3) node[tre] (a) {};
\draw(6,3) node[tre] (b) {};
\draw(3,6) node[tre] (r) {};
 \draw[->] (a)--(1);
  \draw[->] (a)--(3);
  \draw[->] (a)--(2);
   \draw[->] (b)--(3);
 \draw[->] (b)--(4);
 \draw[->] (r)--(a);
 \draw[->] (r)--(b);
            \end{tikzpicture}
                  \begin{tikzpicture}[thick,>=stealth,scale=0.45]
              \draw(0,0) node{\ };
              \draw(0,3) node  {$\Longrightarrow $};
              \draw(0,4) node  {\footnotesize $H(3;1,4)$};
              \draw(0,6) node{\ };
 \end{tikzpicture}
 \begin{tikzpicture}[thick,>=stealth,scale=0.45]
                \draw(0,0) node[tre] (1) {}; \etqsm 1
                   \draw(2,0) node[tre] (2) {}; \etqsm 2
                \draw(6,0) node[tre] (4) {}; \etqsm 4
\draw(1,3) node[tre] (a) {};
\draw(6,3) node[tre] (b) {};
\draw(3,6) node[tre] (r) {};
 \draw[->] (a)--(1);
  \draw[->] (a)--(2);
 \draw[->] (b)--(4);
 \draw[->] (r)--(a);
 \draw[->] (r)--(b);
            \end{tikzpicture}\\[2ex]
                  \begin{tikzpicture}[thick,>=stealth,scale=0.45]
              \draw(0,0) node{\ };
              \draw(0,3) node  {$\Longrightarrow $};
              \draw(0,4) node  {\footnotesize $T(2;1)$};
              \draw(0,6) node{\ };
 \end{tikzpicture}
 \begin{tikzpicture}[thick,>=stealth,scale=0.45]
                \draw(0,0) node[tre] (1) {}; \etqsm 1
                \draw(4,0) node[tre] (4) {}; \etqsm 4
\draw(0,3) node[tre] (a) {};
\draw(4,3) node[tre] (b) {};
\draw(2,6) node[tre] (r) {};
 \draw[->] (a)--(1);
 \draw[->] (b)--(4);
 \draw[->] (r)--(a);
 \draw[->] (r)--(b);
            \end{tikzpicture}
                 \begin{tikzpicture}[thick,>=stealth,scale=0.45]
              \draw(0,0) node{\ };
              \draw(0,3) node  {$\Longrightarrow $};
              \draw(0,4) node  {\footnotesize $U(1)$};
              \draw(0,6) node{\ };
 \end{tikzpicture}
 \begin{tikzpicture}[thick,>=stealth,scale=0.45]
                \draw(4,0) node[tre] (4) {}; \etqsm 4
\draw(0,3) node[tre] (1) {};\etqsm 1
\draw(4,3) node[tre] (b) {};
\draw(2,6) node[tre] (r) {};
 \draw[->] (b)--(4);
 \draw[->] (r)--(1);
 \draw[->] (r)--(b);
            \end{tikzpicture}
                 \begin{tikzpicture}[thick,>=stealth,scale=0.45]
              \draw(0,0) node{\ };
              \draw(0,3) node  {$\Longrightarrow $};
              \draw(0,4) node  {\footnotesize $U(4)$};
              \draw(0,6) node{\ };
 \end{tikzpicture}
 \begin{tikzpicture}[thick,>=stealth,scale=0.45]
\draw(0,3) node[tre] (1) {};\etqsm 1
\draw(4,3) node[tre] (4) {}; \etqsm 4
\draw(2,6) node[tre] (r) {};
 \draw[->] (r)--(1);
 \draw[->] (r)--(4);
\draw(0,0) node{\ };
\end{tikzpicture}

\end{center}
\caption{\label{fig:seqred} 
A sequence of reductions.}
\end{figure}
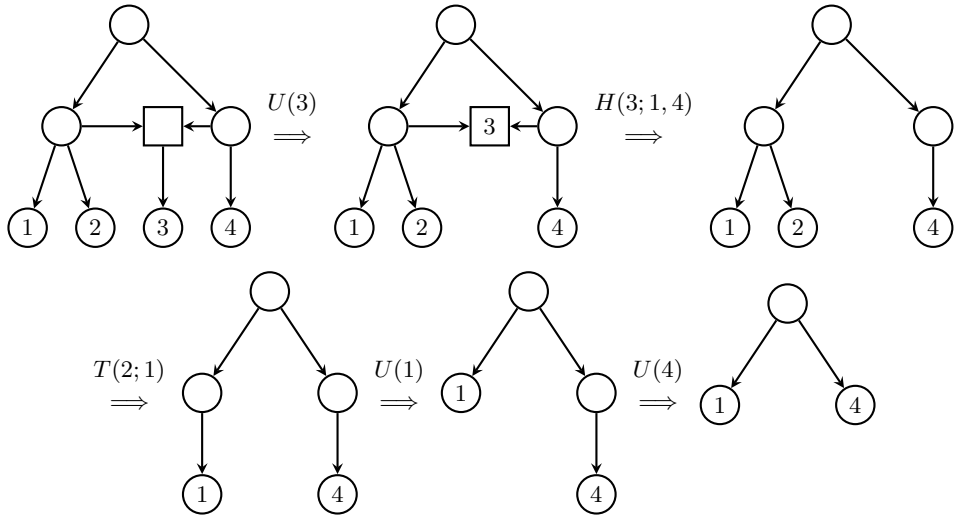

\section{Path lengths vectors for fully resolved networks}
\label{sec:pl-tctc}

Let $N$ be a hybridization network on $S=\{1,\dots,n\}$.  For
every pair of leaves ${i},{j}$ of $N$, let $\ell_N(i,j)$ and
$\ell_N(j,i)$ be the distance from $[i,j]$ to $i$ and to $j$,
respectively.

\begin{definition}
The
\emph{LCSA-path length} between two leaves $i$ and $j$ in $N$ is
$$
L_N(i,j)=\ell_N(i,j)+\ell_N(j,i).
$$
The \emph{LCSA-path lengths vector} of $N$ is
$$
L(N)=\big(L_N(i,j)\big)_{1\leq i<j\leq n}\in \NN^{n(n-1)/2},
$$
with its entries ordered lexicographically in $(i,j)$.
\end{definition}

Notice that $L_N(i,j)=L_N(j,i)$, for every pair of leaves $i,j\in S$.

If $N$ is a phylogenetic tree, the LCSA-path length between two leaves is the
path length between them as defined in \S \ref{sec:trees}, and therefore the vectors $L(N)$ defined therein and here are the same.  But,
contrary to what happens in phylogenetic trees, the LCSA-path length
between two leaves $i$ and $j$ in a hybridization network need not be the smallest sum of the
distances from a common ancestor of $i$ and $j$ to these leaves (that is, the distance between these leaves in the undirected graph associated to the network).

\begin{example}
\label{ex:D}
Consider the TCTC-network $N$ depicted in Fig.~\ref{fig:exnet1}.
Table \ref{table:exnet1} gives, in its upper triangle, the LCSA of
every pair of different leaves, and in its lower triangle, the
LCSA-path length between every pair of different leaves.

Notice that, in this network, $[3,5]=r$, because the root is the only common
ancestor of 3 and 5 that is strict ancestor of some of them, and hence
$L_N(3,5)=8$, but $e$ is a common ancestor of both leaves and the
length of both paths $e\pathgr 3$ and $e\pathgr 5$ is 3. Similarly,
$f$ is also a common ancestor of both leaves and the length of both paths $f \pathgr 3$ and $f \pathgr 5$ is 3.
 This is
an example of LCSA-path length between two leaves that is
largest than the smallest sum of the distances from a common ancestor
of these leaves to each one of them.

\end{example}

\begin{figure}[htb]
\begin{center}
            \begin{tikzpicture}[thick,>=stealth,scale=0.6]
                \draw(0,0) node[tre] (1) {}; \etq 1
                   \draw(2,0) node[tre] (2) {}; \etq 2
                \draw(4,0) node[tre] (3) {}; \etq 3
                \draw(6,0) node[tre] (4) {}; \etq 4
                \draw(8,0) node[tre] (5) {}; \etq 5
                     \draw(10,0) node[tre] (6) {}; \etq 6
                \draw(8,0) node[tre] (5) {}; \etq 5                
\draw(4,2) node[hyb] (A) {};  \etq A
\draw(8,2) node[hyb] (B) {};  \etq B
\draw(0,5) node[tre] (a) {};  \etq a
\draw(2,5) node[tre] (b) {};  \etq b
\draw(6,5) node[tre] (c) {};  \etq c
\draw(10,5) node[tre] (d) {};  \etq d
\draw(1,7) node[tre] (e) {};  \etq e
\draw(8,7) node[tre] (f) {};  \etq f
\draw(5,8) node[tre] (r) {};  \etq r
 \draw[->] (a)--(1);
  \draw[->] (a)--(B);
  \draw[->] (b)--(2);
   \draw[->] (b)--(A);
 \draw[->] (c)--(A);
 \draw[->] (A)--(3);
 \draw[->] (c)--(4);
 \draw[->] (d)--(B);
 \draw[->] (B)--(5);
 \draw[->] (d)--(6);
 \draw[->] (e)--(a);
 \draw[->] (e)--(b);
 \draw[->] (f)--(c);
 \draw[->] (f)--(d);
 \draw[->] (r)--(e);
 \draw[->] (r)--(f);
            \end{tikzpicture}
\end{center}
\caption{\label{fig:exnet1} 
The network $N$ in Example~\ref{ex:D}.}
\end{figure}
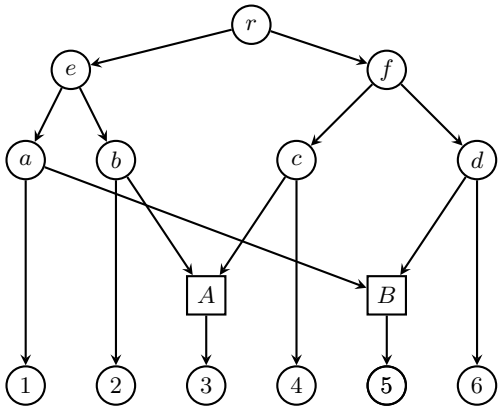

\begin{table}[htd]
\label{table:exnet1}
\caption{For every $1\leq i<j\leq 6$, the entry $(i,j)$ of this table
is $[i,j]$, and the entry $(j,i)$ is $L_N(i,j)$, with $N$ the
network in Fig.~\ref{fig:exnet1}.}
\begin{center}
\begin{tabular}{c||c|c|c|c|c|c|}
&\ 1 \quad & \ 2 \quad & \ 3 \quad & \ 4 \quad & \ 5 \quad & \ 6\quad \\
\hline
\hline
1 & & $e$ & $e$ & $r$ & $a$ & $r$ \\
\hline
2 & 4 & & $b$ & $r$ & $e$ & $r$\\
\hline 
3 & 5 & 3 & & $c$ & $r$ & $f$ \\
\hline 
4 & 6 & 6  & 3 & & $f$ & $f$ \\
\hline
5 & 3 & 5 & 8 & 5 & & $d$\\
\hline
6 & 6 & 6 & 5 & 4 & 3 &  \\
\hline
\end{tabular}
\end{center}
\end{table}%

In a \emph{fully resolved} phylogeny with reticulation events, every non extant species should have two direct descendants, and every
reticulation event should involve two parent species, as such an event corresponds always to the  exchange of 
genetic information between two parents: as 
Semple points out \cite{semple:07},
hybrid nodes with in-degree greater than 2 actually represent ``an uncertainty of the exact order of `hybridization'.'' Depending on whether we
use hybridization or phylogenetic networks to model phylogenies, we
distinguish between:
\begin{itemize}
\item \emph{Fully resolved hybridization networks}: hybridization
networks with all their nodes of types $(0,2)$, $(1,0)$, $(1,2)$,
$(2,0)$, or $(2,2)$.

\item \emph{Fully resolved phylogenetic networks}: phylogenetic
networks with all their nodes of types $(0,2)$, $(1,0)$, $(1,2)$, or
$(2,1)$.
\end{itemize}
To simplify the language, we shall say that a hybridization network is
\emph{quasi-binary} when all its nodes are of types $(0,2)$, $(1,0)$,
$(1,2)$, $(2,0)$, $(2,1)$, or $(2,2)$.  These quasi-binary networks
include as special cases the fully resolved hybridization and
phylogenetic networks.

Our main result in this section establishes that the LCSA-path lengths
vectors separate fully resolved (hybridization or phylogenetic)
TCTC-networks, thus generalizing Proposition \ref{prop:nod-bintree}
from trees to networks.  To prove this result, we shall use the same
strategy as the one developed in \cite{cardona.ea:sbTSTC:2008} or
\cite{comparison2} to prove that the metrics introduced therein were
indeed metrics: algebraic induction based on reductions.  Now, we cannot use the reductions defined in the last section as they stand, because they may generate elementary nodes that are forbidden in fully resolved networks. Instead, we shall use certain suitable combinations of them that always reduce in one the number of leaves.

So, consider the following reduction procedures for quasi-binary TCTC
networks $N$ with $n$ leaves:

\begin{enumerate}
\item[($\mathbf{R}$)] Let $i,j$ be two sibling tree leaves of $N$.
The \emph{$R(i;j)$ reduction} of $N$ is the quasi-binary TCTC-network
$N_{R(i;j)}$ obtained by applying first the $T(i;j)$ reduction to $N$
and then the $U(j)$ reduction to the resulting network.  The final
result is that the leaves $i$ and $j$ are removed, together with their
incoming arcs, and then their former common parent, which now has
become a leaf, is labeled with $j$; cf.~Fig.~\ref{fig:R-red}.

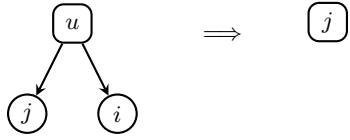
\begin{figure}[htb]
\centering
  \begin{tikzpicture}[thick,>=stealth,scale=0.6]
    \draw (0,0) node[indef] (u) {}; \etqsm{u}
    \draw (-1,-2) node[tre] (j) {}; \etqsm {j}
    \draw (1,-2) node[tre] (i) {}; \etqsm {i}
    \draw [->](u)--(i);
    \draw [->](u)--(j);
  \end{tikzpicture}
  \qquad
               \begin{tikzpicture}[thick,>=stealth,scale=0.6]
              \draw(0,-2) node{\ };
              \draw(0,0) node  {$\Longrightarrow$};
 \end{tikzpicture}
  \qquad
  \begin{tikzpicture}[thick,>=stealth,scale=0.6]
    \draw (0,0.3) node[indef] (u) {}; \draw (u) node {\small $j$};
    \draw (0,-2) node  {\ };  
  \end{tikzpicture}

\caption{\label{fig:R-red} 
The $R(i;j)$-reduction.}
\end{figure}

\item[($\mathbf{H}_0$)]
Let $i$ be a hybrid leaf, let $v_1$ and $v_2$ be its parents and
assume that the other children of these parents are tree leaves $j_1$
and $j_2$, respectively.  The \emph{$H_0(i;j_1,j_2)$ reduction} of $N$
is the quasi-binary TCTC-network $N_{H_0(i;j_1,j_2)}$ obtained by
applying first the reduction $H(i;j_1,j_2)$ to $N$ and then the
reductions $U(j_1)$ and $U(j_2)$ to the resulting network.  The
overall effect is that the hybrid leaf $i$ and the tree leaves $j_1,
j_2$ are removed, together with their incoming arcs, and then the
former parents $v_1,v_2$ of $j_1$ and $j_2$ are labeled with $j_1$ and
$j_2$, respectively; cf.~Fig.~\ref{fig:H0-red}.

\begin{figure}[htb]
\centering
  \begin{tikzpicture}[thick,>=stealth,scale=0.6]
    \draw (0,0) node[hyb] (i) {}; \etqfn{i}
    \draw (-2,0.5) node[indef] (v_1) {}; \etqfn{v_1}
    \draw (-2,-2) node[tre] (j_1) {}; \etqfn{j_1}
    \draw (2,0.5) node[indef] (v_2) {}; \etqfn{v_2}
    \draw (2,-2) node[tre] (j_2) {}; \etqfn{j_2}
    \draw[->](v_1) to (j_1);
    \draw[->](v_1) to  (i);
    \draw[->](v_2) to (j_2);
    \draw[->](v_2) to  (i);    
  \end{tikzpicture} 
  \qquad
               \begin{tikzpicture}[thick,>=stealth,scale=0.6]
              \draw(0,-2) node{\ };
              \draw(0,0.7) node  {$\Longrightarrow$};
 \end{tikzpicture}
\qquad
  \begin{tikzpicture}[thick,>=stealth,scale=0.6]
              \draw(0,-2) node{\ };
    \draw (-2,0.8) node[indef] (j_1) {}; \etqfn{j_1}
    \draw (2,0.8) node[indef] (j_2) {}; \etqfn{j_2}
  \end{tikzpicture}

\caption{\label{fig:H0-red} 
The $H_0(i;j_1,j_2)$-reduction.}
\end{figure}
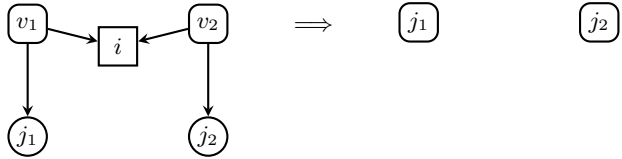

\item[($\mathbf{H}_1$)]
Let $A$ be a hybrid node with only one child $i$, that is a tree node.
Let $v_1$ and $v_2$ be the parents of $A$ and assume that the other
children of these parents are tree leaves $j_1$ and $j_2$,
respectively.  The \emph{$H_1(i;j_1,j_2)$ reduction} of $N$ is the
TCTC-network $N_{H_1(i;j_1,j_2)}$ obtained by applying first the
reduction $U(i)$ to $N$, followed by the reduction $H_0(i;j_1,j_2)$ to
the resulting network.  The overall effect is that the leaf $i$, its
parent $A$ and the leaves $j_1, j_2$ are removed, together with their
incoming arcs, and then the former parents $v_1,v_2$ of $j_1$ and
$j_2$ are labeled with $j_1$ and $j_2$, respectively;
cf.~Fig.~\ref{fig:H1-red}.

\begin{figure}[htb]
\centering
  \begin{tikzpicture}[thick,>=stealth,scale=0.6]
    \draw (0,0) node[hyb] (A) {}; \etqfn{A}
    \draw (0,-2) node[tre] (i) {}; \etqfn{i}
    \draw (-2,0.5) node[indef] (v_1) {}; \etqfn{v_1}
    \draw (-2,-2) node[tre] (j_1) {}; \etqfn{j_1}
    \draw (2,0.5) node[indef] (v_2) {}; \etqfn{v_2}
    \draw (2,-2) node[tre] (j_2) {}; \etqfn{j_2}
    \draw[->](A)--(i);
    \draw[->](v_1) to (j_1);
    \draw[->](v_1) to  (A);
    \draw[->](v_2) to (j_2);
    \draw[->](v_2) to  (A);    
  \end{tikzpicture} 
  \qquad
               \begin{tikzpicture}[thick,>=stealth,scale=0.6]
              \draw(0,-2) node{\ };
              \draw(0,0.7) node  {$\Longrightarrow$};
 \end{tikzpicture}
\qquad
  \begin{tikzpicture}[thick,>=stealth,scale=0.6]
              \draw(0,-2) node{\ };
    \draw (-2,0.8) node[indef] (j_1) {}; \etqfn{j_1}
    \draw (2,0.8) node[indef] (j_2) {}; \etqfn{j_2}
  \end{tikzpicture} 

\caption{\label{fig:H1-red} 
The $H_1(i;j_1,j_2)$-reduction.}
\end{figure}
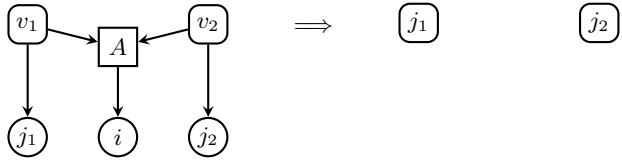
\end{enumerate}
We use $\textrm{H}_0$ and $\textrm{H}_1$ instead of  $\textrm{H}$ and U because, for our purposes in this section, it has to be possible to decide whether or not we can apply a given reduction to a given fully resolved network $N$ from the knowledge of $L(N)$, and this cannot be done for the U reduction, while, as we shall see below, it is possible for $\textrm{H}_0$ and  $\textrm{H}_1$.

$\textrm{H}_0$ reductions cannot be applied to fully resolved phylogenetic networks (they don't have hybrid leaves) and $\textrm{H}_1$ reductions cannot be applied to fully resolved hybridization networks
(they don't have out-degree 1 hybrid nodes). The application of an R or an $\textrm{H}_0$ reduction to a fully resolved TCTC hybridization network is again a fully resolved TCTC hybridization network, and the application of an R or an $\textrm{H}_1$ reduction to a fully resolved TCTC phylogenetic network is again a fully resolved TCTC phylogenetic network. 

We shall call the inverses of the R, H${}_0$ and H${}_1$ reduction
procedures, respectively, the $\textrm{R}^{-1}$, $\textrm{H}_0^{-1}$
and $\textrm{H}_1^{-1}$ \emph{expansions}, and we shall denote them by
$R^{-1}(i;j)$, $H_1^{-1}(i;j_1,j_2)$ and $H_0^{-1}(i;j_1,j_2)$.  More
specifically, for every quasi-binary TCTC-network $N$ with no leaf
labeled $i$:
\begin{itemize}

\item the expansion $R^{-1}(i;j)$ can be applied to $N$ if it has a leaf labeled $j$, and
the resulting network $N_{R^{-1}(i;j)}$ is obtained by unlabeling the
leaf $j$ and adding to it two leaf tree children labeled with $i$ and
$j$;

\item the expansion $H_0^{-1}(i;j_1,j_2)$ can be applied to
$N$ if it has a pair of leaves labeled $j_1,j_2$, and the resulting network $N_{H_0^{-1}(i;j_1,j_2)}$ is obtained
by adding a new hybrid leaf labeled with $i$, and then, for each $l=1,2$,
unlabeling the leaf $j_l$ and adding to it a new tree leaf child labeled
with $j_l$ and an arc to $i$.

\item the expansion $H_1^{-1}(i;j_1,j_2)$  can be applied to
$N$ if it has a pair of leaves labeled $j_1,j_2$, and the resulting network $N_{H_1^{-1}(i;j_1,j_2)}$ is obtained
by adding a new node $A$, a tree leaf child $i$ to it, and then, for
each $l=1,2$, unlabeling the leaf $j_l$ and adding to it a new tree leaf
child labeled with $j_l$ and an arc to $A$.
\end{itemize}
A  $\textrm{R}^{-1}(i;j)$ expansion of a quasi-binary TCTC-network is always a  quasi-binary TCTC-network, but an
$\textrm{H}_0^{-1}(i;j_1,j_2)$ or an
$\textrm{H}_1^{-1}(i;j_1,j_2)$ expansion of a quasi-binary TCTC-network, while still being always quasi-binary and tree child, needs not be time consistent: for instance, 
the leaves $j_1$ and $j_2$ could be a hybrid leaf and a tree sibling of it.
Moreover, we have the following result, which is a direct consequence of Lemma~\ref{lem:iso-red} and we state it for further reference.

\begin{lemma}
\label{lem:iso-red-bin}
Let $N$ and $N'$ be two quasi-binary TCTC-networks.  If $N\cong N'$,
then the result of applying to both $N$ and $N'$ the same R reduction
(respectively, $\textrm{H}_0$ reduction, $\textrm{H}_1$ reduction,
$\textrm{R}^{-1}$ expansion, $\textrm{H}_0^{-1}$ expansion, or
$\textrm{H}_1^{-1}$ expansion) is again two isomorphic networks.

Moreover, if we apply an R reduction (respectively, $\textrm{H}_0$
reduction or $\textrm{H}_1$ reduction) to a quasi-binary TCTC-network
$N$, and then we apply to the resulting network the inverse
$\textrm{R}^{-1}$ expansion (respectively, $\textrm{H}_0^{-1}$
expansion or $\textrm{H}_1^{-1}$ expansion), we obtain a quasi-binary
TCTC-network isomorphic to $N$.  \qed
\end{lemma}

We have moreover the following result.

\begin{proposition}\label{thm-reduction-possible-bin}
Let $N$ be a quasi-binary TCTC-network with more than one leaf.  Then,
at least one R, $\textrm{H}_0$, or $\textrm{H}_1$ reduction can be
applied to $N$.
\end{proposition}

\begin{proof}
If $N$ contains some internal node with two tree leaf children $i$ and
$j$, then the reduction $R(i;j)$ can be applied.  If $N$ does not
contain any node with two tree leaf children, then, by Lemma
\ref{lem:key}, it contains a hybrid node $v$ that is either a leaf
(say, labeled with $i$) or it has only one child, which is a tree leaf
(say, labeled with $i$), and such that all siblings of $v$ are leaves or
hybrid nodes.  Now, the quasi-binarity of $N$ and  the tree child condition entail that $v$ has two
parents, that each one of them  has exactly
one child other than $v$, and that this second child is a
tree node.  So, $v$ has exactly two siblings, and they are tree
leaves, say $j_1$ and $j_2$.  Then, the reduction $H_0(i;j_1,j_2)$ (if
$v$ is a leaf) or $H_1(i;j_1,j_2)$ (if $v$ is not a leaf) can be
applied.
 \end{proof}
 
\begin{corollary}\label{cor-reduction-possible-bin}
\begin{enumerate}[(a)]
\item If $N$ is a fully resolved TCTC hybridization network with more
than one leaf, then, at least one R or $\textrm{H}_0$ reduction can be
applied to it.

\item If $N$ is a fully resolved TCTC phylogenetic network with more
than one leaf, then, at least one R or $\textrm{H}_1$ reduction can be
applied to it.\qed
\end{enumerate}
\end{corollary}

We shall prove now that the application conditions for the reductions
introduced above can be read from the LCSA-path lengths vector of a
fully resolved TCTC-network and that they modify in a specific way
the LCSA-path lengths of the network which they are applied to.  This
will entail that if two fully resolved (hybridization or phylogenetic) TCTC-networks have the same
LCSA-path lengths vectors, then the same reductions can be applied to
both networks and the resulting fully resolved
 TCTC-networks still have
the same LCSA-path lengths vectors.  This will be the basis of the proof by induction on the number of leaves that two TCTC hybridization or phylogenetic networks with  the same
LCSA-path lengths vectors are always isomorphic.

\begin{lemma}
\label{lem:qb-2}
Let $i,j$ be two leaves of a quasi-binary TCTC-network $N$.  Then, $i$
and $j$ are siblings if, and only if, $L_N(i,j)=2$.
\end{lemma}

\begin{proof}
If $L_N(i,j)=2$, then the paths $[i,j]\pathgr i$ and $[i,j]\pathgr j$
have length 1, and therefore $[i,j]$ is a parent of $i$ and $j$.
Conversely, if $i$ and $j$ are siblings and $u$ is a parent in
common of them, then, by the quasi-binarity of $N$, they are the only
children of $u$, and by the tree-child condition, one of them, say
$i$, is a tree node.  But then, $u$ is a strict ancestor of $i$, an
ancestor of $j$, and no proper descendant of $u$ is an ancestor
of both $i$ and $j$.  This implies that $u=[i,j]$ and hence that
$L_N(i,j)=2$.
\end{proof}

\begin{lemma}
\label{lem:qb-R}
Let $N$ be a quasi-binary TCTC-network on a set $S$ of taxa.

\begin{enumerate}[(1)]
\item The reduction $R(i;j)$ can be applied to $N$ if, and only if,
$L_N(i,j)=2$ and, for every $k\in S\setminus\{i,j\}$,
$L_N(i,k)=L_N(j,k)$.

\item If the reduction $R(i;j)$ can be applied to $N$, then
$$
\begin{array}{l}
L_{N_{R(i;j)}}(j,k)=L_N(j,k)-1\quad\mbox{for every $k\in S\setminus\{i,j\}$}\\
L_{N_{R(i;j)}}(k,l)=L_N(k,l)\quad \mbox{for every $k,l\in S\setminus\{i,j\}$}
\end{array}
$$
\end{enumerate}
\end{lemma}

\begin{proof}
As far as (1) goes, $R(i;j)$ can be applied to $N$ if, and only if,
the leaves $i$ and $j$ are siblings and of tree type.  Now, if $i$ and
$j$ are two tree sibling leaves and $u$ is their parent, then on the
one hand, $L_N(i,j)=2$ by Lemma \ref{lem:qb-2}, and on the other hand, since, by Lemma
\ref{lem:LCSA}, $[i,k]=[u,k]=[j,k]$ for every leaf $k\neq i,j$, we have
that 
$$
\begin{array}{l}
\ell_N(i,k)=\ell_N(j,k)=1+\mbox{distance from $[u,k]$ to $u$}\\
\ell_{N}(k,i)=\ell_{N}(k,j)=\mbox{distance from $[u,k]$ to $k$}
\end{array}
$$
and therefore
$L_N(i,k)=L_N(j,k)$ for every $k\in S\setminus\{i,j\}$.

Conversely, assume that $L_N(i,j)=2$ and that $L_N(i,k)=L_N(j,k)$ for
every $k\in S\setminus\{i,j\}$.  The fact that $L_N(i,j)=2$ implies
that $i$ and $j$ share a parent $u$.  If one of these leaves, say $i$,
is hybrid, then the tree child condition implies that the other, $j$,
is of tree type. Let now $v$ be the other parent of $i$ and 
 $k$ a tree descendant leaf of $v$, and let $h$ be
the length of the unique path $v\pathgr k$.  Then $v$ is a strict
ancestor of $k$ and an ancestor of $i$, and no proper  tree descendant
of $v$ can possibly be an ancestor of $i$: otherwise, there would exist a path from a proper tree descendant  of $v$ to $u$, and then the time consistency property would forbid $u$ and $v$ to have a hybrid child in common.
Therefore $v=[i,k]$ and
$L_N(i,k)=h+1$.  Now, the only possibility for the equality $L_N(j,k)=h+1$ to hold is  that some
intermediate node in the path $v\pathgr k$ is an ancestor of the only
parent $u$ of $j$, which, as we have just seen, is impossible.  This leads to a contradiction, which shows that
$i$ and $j$ are both tree sibling leaves.  This finishes the proof of
(1).

As far as (2) goes, in $N_{R(i;j)}$ we remove the leaf $i$ and we
replace the leaf $j$ by its parent.  By Lemma \ref{lem:LCSA}, this
does not modify the LCSA $[j,k]$ of $j$ and any other remaining leaf
$k$, and since we have shortened in 1 any path ending in $j$, we
deduce that $L_{N_{R(i;j)}}(j,k)=L_N(j,k)-1$ for every $k\in
S\setminus\{i,j\}$.  On the other hand, for every $k,l\in
S\setminus\{i,j\}$, the reduction  $R(i;j)$ has
affected neither the LCSA $[k,l]$ of $k$ and $l$, nor the paths $[k,l]\pathgr k$ or $[k,l]\pathgr l$, which implies that $L_{N_{R(i;j)}}(k,l)=L_N(k,l)$ 
\end{proof}


\begin{lemma}
\label{lem:qb-H0}
Let $N$ be a fully resolved  TCTC hybridization network on a set $S$ of taxa.

\begin{enumerate}[(1)]
\item The reduction $H_0(i;j_1,j_2)$ can be applied to $N$ if, and only
if, $L_N(i,j_1)=L_N(i,j_2)=2$.

\item If the reduction $H_0(i;j_1,j_2)$ can be applied to $N$, then
$$
\begin{array}{l}
L_{N_{H_0(i;j_1,j_2)}}(j_1,j_2)=L_N(j_1,j_2)-2\\
L_{N_{H_0(i;j_1,j_2)}}(j_1,k)=L_N(j_1,k)-1\quad \mbox{for every $k\in S\setminus\{i,j_1,j_2\}$}\\
L_{N_{H_0(i;j_1,j_2)}}(j_2,k)=L_N(j_2,k)-1\quad \mbox{for every $k\in S\setminus\{i,j_1,j_2\}$}\\
L_{N_{H_0(i;j_1,j_2)}}(k,l)=L_N(k,l)\quad \mbox{for every $k,l\in S\setminus\{i,j_1,j_2\}$}
\end{array}
$$
\end{enumerate}
\end{lemma}

\begin{proof}
As far as (1) goes, the reduction $H_0(i;j_1,j_2)$ can be applied to
$N$ if, and only if, $i$ is a hybrid sibling of the tree leaves $j_1$
and $j_2$.  If this last condition happens, then $L_N(i,j_1)=2$ and
$L_N(i,j_2)=2$ by Lemma \ref{lem:qb-2}.  Conversely,
$L_N(i,j_1)=L_N(i,j_2)=2$ implies that $i,j_1$ and $i,j_2$ are pairs
of sibling leaves.  Since no node of $N$ can have more than 2 children, and
at least one of its children must be of tree type, this implies that
$i$ is a hybrid node (with two different parents), and $j_1$ and $j_2$
are tree nodes. 

As far as (2) goes, the tree leaves $j_1$ and $j_2$ are replaced by
their parents.  By Lemma \ref{lem:qb-2}, this does not affect any LCSA
and it only shortens in 1 the paths ending in $j_1$ or $j_2$.  Thus,
the $H_0(i;j_1,j_2)$ reduction does not affect the LCSA-path length between any
pair of remaining leaves other than $j_1$ and $j_2$, it shortens in 1
the LCSA-path length between $j_1$ or $j_2$ and any remaining leaf
other than $j_1$ or $j_2$, and it shortens in 2 the LCSA-path length
between $j_1$ and $j_2$.
\end{proof}

\begin{lemma}
\label{lem:qb-H1}
Let $N$ be a fully resolved  TCTC phylogenetic network on a set $S$ of taxa.

\begin{enumerate}[(1)]
\item The reduction $H_1(i;j_1,j_2)$ can be applied to $N$ if, and only
if,
\begin{itemize}
\item $L_N(i,j_1)=L_N(i,j_2)=3$, 
\item $L_N(j_1,j_2)\geq 4$, 
\item  if
$L_N(j_1,j_2)=4$, then $L_N(j_1,k)=L_N(j_2,k)$ for every $k\in
S\setminus\{j_1,j_2,i\}$.
\end{itemize}

\item If the reduction $H_1(i;j_1,j_2)$ can be applied to $N$, then
$$
\begin{array}{l}
L_{N_{H_1(i;j_1,j_2)}}(j_1,j_2)=L_N(j_1,j_2)-2\\
L_{N_{H_1(i;j_1,j_2)}}(j_1,k)=L_N(j_1,k)-1\ \mbox{for every $k\in S\setminus\{i,j_1,j_2\}$}\\
L_{N_{H_1(i;j_1,j_2)}}(j_2,k)=L_N(j_2,k)-1\ \mbox{for every $k\in S\setminus\{i,j_1,j_2\}$}\\
L_{N_{H_1(i;j_1,j_2)}}(k,l)=L_N(k,l)\ \mbox{for every $k\in S\setminus\{i,j_1,j_2\}$}
\end{array}
$$
\end{enumerate}
\end{lemma}

\begin{proof}
As far as (1) goes, the reduction $H_1(i;j_1,j_2)$ can be applied to
$N$ if, and only if, $j_1$ and $j_2$ are tree leaves that are not siblings and they share a sibling hybrid node that has  the tree leaf $i$ as its only child.
Now, if this application condition for $H_1(i;j_1,j_2)$ is satisfied, then $L_N(i,j_1)=3$,
because the parent of $j_1$ is an ancestor of $i$, a strict ancestor
of $j_1$, and clearly no proper descendant of it is an ancestor
of $i$ and $j_1$; by a similar reason, $L_N(i,j_2)=3$.  Moreover, since $j_1$
and $j_2$ are not sibling, $L_N(j_1,j_2)\geq 3$.  But if
$L_N(j_1,j_2)= 3$, then there would exist an arc from the parent of
$j_1$ to the parent of $j_2$, or vice versa, which would entail a node
of out-degree 3 that cannot exist in the fully resolved network $N$.
Therefore, $L_N(j_1,j_2)\geq 4$.  Finally, if $L_N(j_1,j_2)=4$, this
means that the parents $x$ and $y$ of $j_1$ and $j_2$ (that are tree
nodes, because they have out-degree 2 and $N$ is a phylogenetic network) are sibling: let $u$ be their
parent in common.  In this case, no leaf other than $j_1,j_2,i$ is a
descendant of $u$, and therefore, for every $k\in S\setminus\{j_1,j_2,i\}$,  
$$
[j_1,k]=[x,k]=[u,k]=[y,k]=[j_2,k]
$$
by Lemma \ref{lem:key},
and thus
$$
\begin{array}{l}
\ell_N(j_1,k)=\ell_N(j_2,k)=2+\mbox{distance from $[u,k]$ to $u$}\\
\ell_{N}(k,j_1)=\ell_{N}(k,j_2)=\mbox{distance from $[u,k]$ to $k$},
\end{array}
$$
 which implies that
$L_N(j_1,k)=L_N(j_2,k)$.


Conversely, assume that $L_N(i,j_1)=L_N(i,j_2)=3$, that
$L_N(j_1,j_2)\geq 4$, and that if $L_N(j_1,j_2)=4$, then
$L_N(j_1,k)=L_N(j_2,k)$ for every $k\in S\setminus\{j_1,j_2,i\}$.  Let
$x$, $y$ and $z$ be the parents of $j_1$, $j_2$ and $i$, respectively.
Notice that these parents are pairwise different (otherwise, the
LCSA-path length between a pair among $j_1,j_2,i$ would be 2).  Moreover,
since $N$ is a phylogenetic network, $j_1$, $j_2$ and $i$ are tree nodes.
Then, 
$L_N(i,j_1)=L_N(i,j_2)=3$ implies that there must exist an arc between
the nodes $x$ and $z$ and an arc between the nodes $y$ and $z$.

Now, if these arcs are $(z,x)$ and $(z,y)$, the node $z$ would have
out-degree 3, which is impossible.  Assume now that $(x,z)$ and
$(z,y)$ are arcs of $N$.  In this case, both $z$ and $x$ have
out-degree 2, which implies (recall that $N$ is a phylogenetic network) that they
are tree nodes.  Then, $x=[j_1,j_2]$ (it is an ancestor of
$j_2$, a strict ancestor of $j_1$, and no proper descendant of
it is an ancestor of $j_1$ and $j_2$) and therefore $L_N(j_1,j_2)=4$.
In this case, we assume that $L_N(j_1,k)=L_N(j_2,k)$ for every $k\in
S\setminus\{j_1,j_2,i\}$.  Now we must distinguish two cases,
depending on the type of node $y$:
\begin{itemize}
\item If $y$ is a tree node, let $p$ be its child other than $j_2$,
and let $k$ be a tree descendant leaf of $p$.  In this case, $[j_1,k]=x$
and $[j_2,k]=y$ (by the same reason why $x$ is $[j_1,j_2]$), and hence
$L_{N}(j_1,k)=L_{N}(j_2,k)+2$, against the assumption
$L_{N}(j_1,k)=L_{N}(j_2,k)$.

\item If $y$ is a hybrid node, let $p$ be its parent other than $z$,
and let $k$ be a tree descendant leaf of $p$ ($k\neq {j_2}$, because $j_2$ is not a tree descendant
of $p$).  In this case, $[j_2,k]=p$
(because $p$ is an ancestor of $j_2$ and a strict ancestor of $k$, and
the time consistency property implies that no intermediate node in the
path $p\pathgr k$ can be an ancestor of $y$).  
Now, if the length of the (only) path
$p\pathgr k$ is $h$, then $L_{N}(j_2,k)=h+2$, and for the equality
$L_{N}(j_1,k)=h+2$ to hold, either the arc $(x,p)$ belongs to $N$, which is impossible because $x$ would have out-degree 3, or a node in
the path $p\pathgr k$ is an ancestor of $x$, which is impossible because of the time consistency property.

\end{itemize}
In both cases we reach a contradiction that implies that the arcs
$(x,z),(z,y)$ do not exist in $N$.  By symmetry, the arcs
$(y,z),(z,x)$ do not exist in $N$, either.
Therefore, the only possibility is that $N$ contains the arcs $(x,z), (y,z)$, that is, that  $z$
is hybrid child of the nodes $x$ and $y$. This finishes the proof of (1).

As far as (2) goes, it is proved as in Lemma \ref{lem:qb-H0}.
\end{proof}

Now we can prove the main results in this section.

\begin{proposition}
\label{prop:D-bin-net-H}
Let $N$ and $N'$ be two fully resolved TCTC hybridization networks on
the same set $S$ of taxa.  Then, $L(N)=L(N')$ if, and only if, $N\cong
N'$.
\end{proposition}

\begin{proof}
The `if' implication is obvious.  We prove the `only if' implication
by induction on the number $n$ of elements of $S$.

The cases $n=1$ and $n=2$ are straightforward, because there exist
only one TCTC-network on $S=\{1\}$ and one TCTC-network on $S=\{1,2\}$:
the one-node graph and the phylogenetic tree with leaves 1,2,
respectively.

Assume now that the thesis is true for fully resolved TCTC hybridization networks with $n$
leaves, and let $N$ and $N'$ be two fully resolved TCTC hybridization networks on the same set $S$ of $n+1$ labels such that $L(N)=L(N')$.   By Corollary \ref{cor-reduction-possible-bin}.(a), an $R(i;j)$ or a $H_0(i;j_1,j_2)$ can be applied to $N$. Moreover, since the possibility of applying one such reduction depends on the LCSA-path lengths vector by Lemmas \ref{lem:qb-R}.(1) and \ref{lem:qb-H0}.(1), and $L(N)=L(N')$, it will be possible to apply the same reduction to $N'$.
So, let $N_1$ and $N_1'$ be the  fully resolved TCTC hybridization networks obtained by applying the same R or $\textrm{H}_0$ reduction to $N$ and $N'$.

From Lemmas \ref{lem:qb-R}.(2) and \ref{lem:qb-H0}.(2) we deduce that
$L(N_1)=L(N_1')$ and hence, by the induction hypothesis, $N_1\cong
N_1'$.  Finally, if we apply to $N_1$ and $N_1'$ the  $\mathrm{R}^{-1}$ or $\mathrm{H}_0^{-1}$ expansion that is
inverse to the reduction applied to $N$ and $N'$, then, by Lemma \ref{lem:iso-red-bin}, we obtain
again $N$ and $N'$ and they are isomorphic.
\end{proof}

A similar argument, using Lemmas \ref{lem:qb-R} and \ref{lem:qb-H1}, proves the following result.

\begin{proposition}
\label{prop:D-bin-net-P}
Let $N$ and $N'$ be two fully resolved TCTC phylogenetic networks on
the same set $S$ of taxa.  Then, $L(N)=L(N')$ if, and only if, $N\cong
N'$.\qed
\end{proposition}

\begin{remark}
\label{rem:noqb}
The LCSA-path lengths vectors do not separate quasi-binary TCTC-networks. Indeed, consider the TCTC-networks $N,N'$ depicted in Fig.~\ref{fig:qb}. They are quasi-binary (but neither fully resolved phylogenetic networks nor fully resolved hybridization networks), and a simple computation shows that
$$
L(N)=L(N')=(3,6,3,3,6,3).
$$
The network $N$ in Fig.~\ref{fig:qb} also shows that Lemma \ref{lem:qb-H1}.(1) is false for quasi-binary hybridization networks.
\end{remark}

\begin{figure}[htb]
\begin{center}
            \begin{tikzpicture}[thick,>=stealth,scale=0.5]
\draw(0,0) node[tre] (1) {};  \etq 1
\draw(3,0) node[hyb] (2) {};  \etq 2
\draw(7,0) node[tre] (3) {};  \etq 3
 \draw(9,0) node[tre] (4) {};  \etq 4
\draw(0,2) node[tre] (a) {};  
\draw(9,2) node[hyb] (A) {};  
\draw(1,6) node[tre] (b) {};  
\draw(5,6) node[tre] (d) {};  
\draw(7,4) node[tre] (c) {};  
\draw(3,8) node[tre] (r) {};  
             \draw[->] (a)--(1);
             \draw[->] (a)--(A);
             \draw[->] (A)--(4);
             \draw[->] (b)--(a);
             \draw[->] (b)--(2);
             \draw[->] (d)--(2);
               \draw[->] (d)--(c);
           \draw[->] (c)--(3);
             \draw[->] (c)--(A);
             \draw[->] (r)--(b);
             \draw[->] (r)--(d);
       \draw(3.5,-1.5) node  {$N$};
            \end{tikzpicture}
\qquad
            \begin{tikzpicture}[thick,>=stealth,scale=0.5]
 \draw(0,0) node[hyb] (1) {};  \etq 1
\draw(3,0) node[tre] (2) {};  \etq 2
 \draw(5,0) node[tre] (3) {};  \etq 3
\draw(7,0) node[tre] (4) {};  \etq 4
             \draw(-1,6) node[tre] (a) {};               
              \draw(3,6) node[tre] (b) {};  
              \draw(3,2) node[tre] (c) {};   
              \draw(5,2) node[hyb] (A) {};  
              \draw(7,2) node[tre] (d) {};   
              \draw(2,8) node[tre] (r) {};   
              \draw[->] (a)--(1);
             \draw[->] (b)--(1);
                     \draw[->] (a)--(d);
     \draw[->] (b)--(c);
             \draw[->] (c)--(2);
             \draw[->] (c)--(A);
             \draw[->] (d)--(A);
                \draw[->] (A)--(3);
          \draw[->] (d)--(4);
             \draw[->] (r)--(a);
             \draw[->] (r)--(b);
  \draw(3,-1.5) node  {$N'$};
            \end{tikzpicture}
       \end{center}
\caption{\label{fig:qb}
These two quasi-binary TCTC-networks have the same LCSA-path length vectors.}
\end{figure}
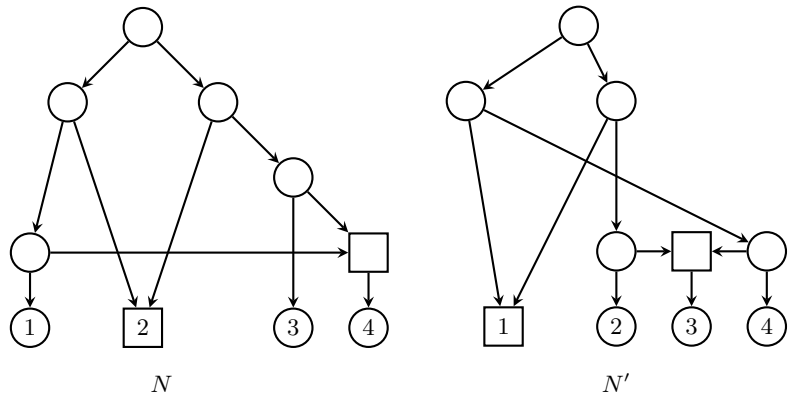

Let $\mathrm{FRH}_n$ (respectively,  $\mathrm{FRP}_n$) denote 
the classes of fully resolved TCTC hybridization (respectively, phylogenetic) networks on $S=\{1,\ldots,n\}$. We have just proved that the mappings
$$
L:\mathrm{FRH}_n\to  \RR^{n(n-1)/2},\quad L:\mathrm{FRP}_n\to  \RR^{n(n-1)/2}
$$
are injective, and therefore they can be used to induce metrics on $\mathrm{FRH}_n$ and $\mathrm{FRP}_n$ from metrics on $\RR^{n(n-1)/2}$.

\begin{proposition}
\label{prop:metric}
For every $n\geq 1$,  let $D$ be any metric on $\RR^{n(n-1)/2}$. The mappings
$d: \mathrm{FRH}_n\times \mathrm{FRH}_n  \to  \RR$ and 
$d: \mathrm{FRP}_n\times \mathrm{FRP}_n   \to  \RR$ defined by
$d(N_1,N_2)= D(L(N_1),L(N_2))$
 satisfy the axioms of metrics up to isomorphisms:
  \begin{enumerate}[(1)]
  \item $d(N_1,N_2)\ge 0$,
  \item $d(N_1,N_2)=0$ if, and only if, $N_1\cong N_2$,
  \item $d(N_1,N_2)=d(N_2,N_1)$,
  \item $d(N_1,N_3)\le d(N_1,N_2)+d(N_2,N_3)$.
  \end{enumerate}
 \end{proposition}

\begin{proof}
Properties (1), (3) and (4) are direct consequences of the
corresponding properties of $D$, while property (2) follows from the
separation axiom for $D$ (which says that $D(M_1,M_2)=0$ if, and only if, $M_1=M_2$)
and Proposition \ref{prop:D-bin-net-H} or \ref{prop:D-bin-net-P}, depending on the case. 
\end{proof}

For instance, using as $D$ the Manhattan distance on $\RR^{n(n-1)/2}$, we obtain the metric 
on $\mathrm{FRH}_n$ or $\mathrm{FRP}_n$
$$
d_1(N_1,N_2)=\sum_{1\leq i<j\leq n}|L_{N_1}(i,j)-L_{N_2}(i,j)|,
$$
and using as $D$ the Euclidean distance we obtain the metric
$$
d_2(N_1,N_2)=\sqrt{\sum_{1\leq i<j\leq n}(L_{N_1}(i,j)-L_{N_2}(i,j))^2}.
$$
These metrics generalize to fully resolved TCTC (hybridization or phylogenetic) networks 
the classical distances for fully resolved phylogenetic trees introduced by Farris \cite{farris:sz69} 
and Clifford \cite{willcliff:taxon71} around 1970.

\section{Splitted path lengths vectors for arbitrary networks}
\label{sec:split}

As we have seen in \S \ref{sec:trees} and Remark~\ref{rem:noqb}, the path lengths vectors do not
separate arbitrary  TCTC-networks. Since to separate arbitrary phylogenetic trees
we splitted the path lengths (Definition \ref{def:split-tree}), we shall use the same strategy in the networks setting.  In this connection, we already proved
in \cite{comparison2} that the 
matrix
$$
\ell(N)=\big(\ell_N(i,j)\big)_{i=1,\ldots,n\atop j=1,\ldots,n} 
$$
separates TCTC \emph{phylogenetic} networks  on $S=\{1,\ldots,n\}$ with tree nodes of arbitrary out-degree and hybrid nodes of arbitrary in-degree. But it is not true for TCTC  hybridization networks, as the following example shows.

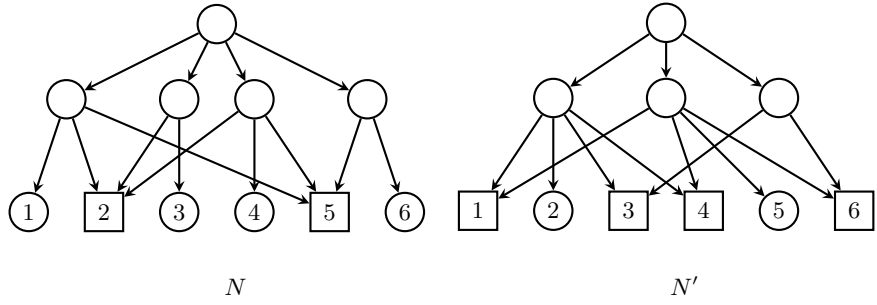
\begin{figure}[htb]
\begin{center}
            \begin{tikzpicture}[thick,>=stealth,scale=0.5]
\draw(0,1) node[tre] (1) {};  \etq 1
\draw(2,1) node[hyb] (2) {};  \etq 2
 \draw(4,1) node[tre] (3) {};  \etq 3
\draw(6,1) node[tre] (4) {};  \etq 4
\draw(8,1) node[hyb] (5) {};  \etq 5
\draw(10,1) node[tre] (6) {};  \etq 6
             \draw(1,4) node[tre] (d) {}; 
              \draw(4,4) node[tre] (e) {}; 
              \draw(9,4) node[tre] (g) {};  
              \draw(6,4) node[tre] (f) {}; 
              \draw(5,6) node[tre] (r) {}; 
             \draw[->] (d)--(1);
             \draw[->] (d)--(2);
             \draw[->] (d)--(5);
             \draw[->] (e)--(2);
             \draw[->] (e)--(3);
             \draw[->] (f)--(2);
               \draw[->] (f)--(4);
           \draw[->] (f)--(5);
             \draw[->] (g)--(5);
             \draw[->] (g)--(6);
             \draw[->] (r)--(d);
             \draw[->] (r)--(e);
             \draw[->] (r)--(f);
             \draw[->] (r)--(g);
  \draw(5.5,-1) node  {$N$};
            \end{tikzpicture}
\quad
            \begin{tikzpicture}[thick,>=stealth,scale=0.5]
 \draw(0,1) node[hyb] (1) {};  \etq 1
\draw(2,1) node[tre] (2) {};  \etq 2
 \draw(4,1) node[hyb] (3) {};  \etq 3
\draw(6,1) node[hyb] (4) {};  \etq 4
\draw(8,1) node[tre] (5) {};  \etq 5
\draw(10,1) node[hyb] (6) {};  \etq 6
             \draw(2,4) node[tre] (c) {}; 
              \draw(5,4) node[tre] (d) {}; 
              \draw(8,4) node[tre] (e) {};  
              \draw(5,6) node[tre] (r) {}; 
              \draw[->] (c)--(1);
             \draw[->] (c)--(2);
                     \draw[->] (c)--(3);
     \draw[->] (c)--(4);
             \draw[->] (d)--(1);
             \draw[->] (d)--(4);
             \draw[->] (d)--(5);
                \draw[->] (d)--(6);
          \draw[->] (e)--(3);
             \draw[->] (e)--(6);
             \draw[->] (r)--(c);
             \draw[->] (r)--(d);
             \draw[->] (r)--(e);
  \draw(5.5,-1) node  {$N'$};
            \end{tikzpicture}
       \end{center}
\caption{\label{fig:calh}
These two hybridization TCTC-networks are such that $\ell_N(i,j)=\ell_{N'}(i,j)$,
for every pair of leaves $i,j$.}
\end{figure}

\begin{example}
Consider the pair of non-isomorphic TCTC-networks $N$ and $N'$
depicted in Fig.~\ref{fig:calh}.  A simple computation shows that
$$
\ell(N)=\ell(N')=
\left(
\begin{array}{cccccc}
    0 & 1 & 2 & 2 & 1 & 2  \\[-1ex]  
    1 & 0 & 1 & 1 &  2 & 2  \\[-1ex]  
    2 & 1 & 0  & 2 & 2 & 2 \\[-1ex] 
    2 & 1 & 2 & 0 & 1 & 2  \\[-1ex] 
    1 & 2 & 2 & 1 & 0 & 1  \\[-1ex] 
    2 & 2 & 2 & 2 & 1 & 0   
 \end{array}
\right)
$$
\end{example}

So, in order to separate arbitrary TCTC-networks we need to add some extra information to the distances  $\ell_N(i,j)$ from LCSAs to leaves. The extra information we shall use is whether the LCSA of each pair of leaves is a strict ancestor of one leaf or the other (or both).
So, for every pair of different leaves $i,j$ of $N$,
let $h_N(i,j)$ be $-1$ if $[i,j]$ is a strict ancestor of $i$ but not of $j$,
 $1$ if $[i,j]$ is a strict ancestor of $j$ but not of $i$,
and  $0$ if $[i,j]$ is a strict ancestor of both $i$ and $j$. Notice that $h_N(j,i)=-h_N(i,j)$.

\begin{definition}
Let $N$ be a hybridization network on the set $S=\{1,\ldots,n\}$. 

For every $i,j\in S$, the \emph{splitted LCSA-path length} from $i$ to $j$ is the ordered  3-tuple
$$
L_N^{s}(i,j)=(\ell_N(i,j), \ell_N(j,i),h_N(i,j)).
$$
The
\emph{splitted LCSA-path lengths vector} of $N$ is
$$
L^{s}(N) = \big(L^{s}_N(i,j)\big)_{1\leq i<j\leq n} \in (\NN\times\NN\times\{-1,0,1\})^{n(n-1)/2}
$$
with its entries ordered lexicographically in $(i,j)$.
\end{definition}

\begin{example}
Consider the quasi-binary TCTC-networks $N$ and $N'$ depicted in Fig.~\ref{fig:qb}.
Then
$$
\begin{array}{rl}
L^s(N) &\!\! =\!\! \big((2,1,-1),(3,3,0),(1,2,-1),(1,2,1),(2,4,0),(1,2,-1)\big)\\
L^s(N') &\!\! =\!\! \big((1,2,1),(2,4,0),(1,2,1),(1,2,-1),(3,3,0),(2,1,1) \big)
\end{array}
$$
\end{example}

\begin{example}
Consider the TCTC-networks $N$ and $N'$ depicted in Fig.~\ref{fig:calh}.
Then
$$
\begin{array}{rl}
L^s(N) &\!\! =\!\! \big((1,1,-1), (2,2,0), (2,2,0), (1,1,-1), (2,2,0), (1,1,1), (1,1,1), \\
& \ (2,2,0), (2,2,0), (2,2,0), (2,2,0), (2,2,0), (1,1,-1), (2,2,0), (1,1,1)\big)\\
L^s(N') &\!\! =\!\! \big((1,1,1), (2,2,0), (2,2,0), (1,1,1), (2,2,0), (1,1,-1), (1,1,-1),\\
& \ (2,2,0), (2,2,0), (2,2,0), (2,2,0), (2,2,0), (1,1,1), (2,2,0), (1,1,-1)\big)
\end{array}
$$
\end{example}

\begin{remark}
If $N$ is a phylogenetic tree on $S$, then $h_N(i,j)=0$ for every  $i,j\in S$.
\end{remark}

We shall prove now that these splitted LCSA-path lengths vectors
separate arbitrary hybridization TCTC-networks. The master plan for proving it is similar to the one used in the proof of Proposition \ref{prop:D-bin-net-H}: induction based on the fact that the application conditions for the reductions introduced in Section \ref{sec:tctc} can be read in the splitted LCSA-path lengths vectors of TCTC-networks and that these reductions modify in a controlled way these vectors.


\begin{lemma}
\label{lem:red-U}
Let $N$ be a  TCTC-network on a set $S$ of taxa.

\begin{enumerate}[(1)]
\item The reduction $U(i)$ can be applied to $N$ if, and only if,
$\ell_N(i,j)\geq 2$ for every $j\in S\setminus\{i\}$.

\item If the reduction $U(i)$ can be applied to $N$,
then
$$
\begin{array}{l}
L_{N_{U(i)}}(i,j)=L_N(i,j)-(1,0,0)
\quad\mbox{for every $j\in S\setminus\{i\}$}\\
L_{N_{U(i)}}(j,k)=L_N(j,k)\quad\mbox{for every $j,k\in S\setminus\{i\}$}
\end{array}
$$
\end{enumerate}
\end{lemma}

\begin{proof}
As far as (1) goes, the reduction $U(i)$ can be applied to $N$ if, and only if,
the leaf $i$ is a tree node and  the only child of its parent.  Let us check now that
this last condition  is equivalent to
$\ell_N(i,j)\geq 2$ for every $j\in S\setminus\{i\}$.
To do this, we distinguish three cases:
\begin{itemize}
\item Assume that $i$ is a tree node and the only child of its parent $x$.
Then, for every $j\in S\setminus\{i\}$, the LCSA of $i$ and $j$ is a proper ancestor of $x$, and therefore $\ell_N(i,j)\geq 2$.

\item Assume that $i$ is a tree node and that it has a sibling $y$. Let $x$ be the parent of $i$ and $y$ and let $j$ be a tree descendant leaf of $y$. Then  $[i,j]=x$, because $x$ is a strict ancestor of $i$, an ancestor of $j$ and clearly no descendant of $x$ is an ancestor of both $i$ and $j$. Therefore, in this case, $\ell_N(i,j)=1$ for this leaf $j$.

\item Assume that $i$ is a hybrid node. Let $x$ be any parent of $i$ and let $j$ be a tree descendant of  $x$. Then, $[i,j]=x$, because $x$ is a strict ancestor of $j$, an ancestor of $i$, and no intermediate node in the unique path $x\pathgr j$ is an ancestor of $i$ (it would violate the time consistency property). Therefore,  in this case, $\ell_N(i,j)=1$ for this leaf $j$, too.
\end{itemize}
Since these three cases cover all possibilities, we conclude that
$i$ is a tree node without siblings if, and only if, $\ell_N(i,j)\geq 2$ for every $j\in S\setminus\{i\}$.
 This finishes the proof of (1).

As far as (2) goes, in $N_{U(i)}$ we replace the tree leaf $i$ by
its parent.  By Lemma \ref{lem:LCSA}, this does not modify any LCSA,
and it only shortens in 1 any path ending in $i$.  Therefore
$$
\begin{array}{l}
\ell_{N_{U(i)}}(i,j)=\ell_{N}(i,j)-1,\ \ell_{N_{U(i)}}(j,i)=\ell_{N}(j,i) \quad\mbox{for every $j\in S\setminus\{i\}$}\\
\ell_{N_{U(i)}}(j,k)=\ell_{N}(j,k),\ \ell_{N_{U(i)}}(k,j)=\ell_{N}(k,j) \quad\mbox{for every $j,k\in
S\setminus\{i\}$}\\
\end{array}
$$
As far as the $h$ component of the splitted LCSA-path lengths goes, notice that a node $u$ is a strict ancestor of a tree leaf $i$ if, and only if, it is a strict ancestor of its parent $x$ (because every path ending in $i$ contains $x$). Therefore, an internal node of $N_{U(i)}$ is a strict ancestor of the leaf $i$ in $N_{U(i)}$ if, and only if, it is a strict ancestor of the leaf $i$ in $N$. On the other hand, replacing a tree leaf without siblings by its only parent does not affect any path ending in another leaf, and therefore an internal node of $N_{U(i)}$ is a strict ancestor of a leaf $j\neq i$ in $N_{U(i)}$ if, and only if, it is a strict ancestor of the leaf $j$ in $N$.

So, by Lemma \ref{lem:LCSA}, the LCSA of a pair of leaves in $N$ and in $N_{U(i)}$ is the same, and we have just proved that this LCSA is a strict ancestor of exactly the same leaves in both networks: this implies that 
$$
\begin{array}{l}
h_{N_{U(i)}}(i,j)=h_{N}(i,j) \quad\mbox{for every $j\in S\setminus\{i\}$}\\
h_{N_{U(i)}}(j,k)=h_{N}(j,k) \quad\mbox{for every $j,k\in
S\setminus\{i\}$}\\
\end{array}
$$
\end{proof}

\begin{lemma}
\label{lem:red-T}
Let $N$ be a  TCTC-network on a set $S$ of taxa.

\begin{enumerate}[(1)]
\item The reduction $T(i;j)$ can be applied to $N$ if, and only if,
$L_N^s(i,j)=(1,1,0)$.

\item If the reduction $T(i;j)$ can be applied to $N$, then 
$$
L^s_{N_{T(i;j)}}(k,l)=L^s_N(k,l)\quad\mbox{for every $k,l\in S\setminus\{i\}$}
$$
\end{enumerate}
\end{lemma}

\begin{proof}
As far as (1) goes, $T(i;j)$ can be applied to $N$ if, and only if,
the leaves $i$ and $j$ are tree nodes and sibling.  Let us prove that this last condition is equivalent to $\ell_N(i,j)=\ell_{N}(j,i)=1$ and $h_N(i,j)=0$.
Indeed, if  the leaves $i$ and $j$ are tree nodes and sibling, then their parent is their LCSA and moreover it is a strict ancestor of both of them, which implies that 
$\ell_N(i,j)=\ell_{N}(j,i)=1$ and $h_N(i,j)=0$.
Conversely, assume that $\ell_N(i,j)=\ell_{N}(j,i)=1$ and $h_N(i,j)=0$.
  The
equalities $\ell_N(i,j)=\ell_{N}(j,i)=1$ imply that $[i,j]$ is a parent of $i$ and $j$,
and $h_N(i,j)=0$ implies that this parent of $i$ and $j$ is a strict ancestor of both of them, and therefore, by Lemma \ref{lem:hyb-nostr}, that $i$ and $j$ are tree nodes. This finishes the proof of (1).

As far as (2) goes, in $N_{T(i;j)}$ we simply remove the leaf $i$
without removing anything else. Therefore, no path ending in a remaining leaf is affected, and as a consequence   no $L^s(k,l)$ with
$k,l\neq i$, is modified.
\end{proof}

\begin{lemma}
\label{lem:red-H}
Let $N$ be a  TCTC-network on a set $S$ of taxa.

\begin{enumerate}[(1)]
\item The reduction $H(i;j_1,\ldots,j_k)$ can be applied to $N$ if,
and only if,
\begin{itemize}
\item $L_N^s(i,j_l)=(1,1,1)$, for every $l\in\{1,\ldots, k\}$.

\item $\ell_N(j_a,j_b)\geq 2$ or $\ell_N(j_b,j_a)\geq 2$ for every $a,b\in \{1,\ldots,k\}$.

\item For every $s\notin \{j_1,\ldots,j_k\}$, if  $\ell_N(i,s)=1$ and $h_N(i,s)=1$, then $\ell_N(j_l,s)=1$ and $h_N(j_l,s)=0$  for some $l\in\{1,\ldots, k\}$.
\end{itemize}

\item If the reduction $H(i;j_1,\ldots,j_k)$ can be applied to $N$, then
$$
L_{N_{H(i;j_1,\ldots,j_k)}}(s,t)=L_N(s,t)\quad\mbox{for every $s,t\in S\setminus\{i\}$}
$$
\end{enumerate}
\end{lemma}

\begin{proof}
As far as (1) goes, $H(i;j_1,\ldots,j_k)$ can be applied to $N$ if,
and only if, $j_1,\ldots,j_k$ are tree leaves that are not sibling of each other, the leaf $i$ is a hybrid sibling of $j_1,\ldots,j_k$, and   the only parents of $i$ are those
of $j_1,\ldots,j_k$.  Now:
\begin{itemize}

\item For each $l=1,\ldots,k$, the condition  $L_N^s(i,j_l)=(1,1,1)$ says that
$i$ and $j_l$ are sibling, and that their parent in common is a strict ancestor of $j_l$ but not of $i$. Using Lemma \ref{lem:hyb-nostr},  we conclude that this condition is equivalent to the fact that $i$ and $j_l$ are sibling, $j_l$ is a tree node,  and $i$ a hybrid node.

\item Assume that $j_1,\ldots, j_k$ are tree leaves, with parents $v_1,\ldots,v_k$, respectively. In this case, the condition 
$\ell_N(j_a,j_b)\geq 2$ or $\ell_N(j_b,j_a)\geq 2$  is equivalent to the fact that $j_a,j_b$  are not sibling.
Indeed, if $j_a$ and $j_b$ are sibling, then $\ell_N(j_a,j_b)=\ell_N(j_b,j_a)=1$.
Conversely, if $j_a$ and $j_b$ are not sibling, then there are two possibilities:
either $v_a$ is an ancestor of $j_b$, but not its parent, in which case $v_a=[j_a,j_b]$ and $\ell_N(j_b,j_a)\geq 2$, or $v_a$ is not an ancestor of $j_b$, in which case $[j_a,j_b]$ is a proper ancestor of $v_a$ and hence $\ell_N(j_a,j_b)\geq 2$.

\item Assume that $j_1,\ldots, j_k$ are tree leaves, with parents $v_1,\ldots,v_k$, respectively, and that $i$ is a hybrid sibling of them. Let us see that the only parents of $i$ are $v_1,\ldots,v_k$ if, and only if,
for every $s\notin \{j_1,\ldots,j_k\}$, $\ell_N(i,s)=1$ and $h_N(i,s)=1$ imply that $\ell_N(j_l,s)=1$ and $h_N(j_l,s)=0$ for
some $l=1,\ldots,k$. 

 Indeed, assume that the only parents of $i$ are $v_1,\ldots,v_k$, and let $s\notin\{j_1,\ldots,j_k\}$ be a leaf such that $\ell_N(i,s)=1$ and $h_N(i,s)=1$. Since $\ell_N(i,s)=1$, some parent of $i$, say $v_l$, is the LCSA of $i$ and $s$, and $h_N(i,s)=1$ implies that $v_l$ is a strict ancestor of $s$. But then
 $v_l$ will be the LCSA of its tree leaf $j_l$ and $s$ and strict ancestor of both of them,
 and thus $\ell_N(j_l,s)=1$ and $h_N(j_l,s)=0$.

Conversely, assume that, for every $s\notin \{j_1,\ldots,j_k\}$, $\ell_N(i,s)=1$ and $h_N(i,s)=1$ imply that $\ell_N(j_l,s)=1$ and $h_N(j_l,s)=0$ for
some $l=1,\ldots,k$. 
Let $v$ 
be a parent of $i$, and let  $s$
be a tree descendant leaf of $v$. Then, 
$v=[i,s]$ ($v$ is a strict ancestor of $s$, an ancestor of $i$, and no intermediate node in the unique path $v\pathgr s$ is an ancestor of $i$, by the time consistency property) and thus
$\ell_N(i,s)=1$; moreover, $h_N(i,s)=1$ by Lemma \ref{lem:hyb-nostr}.
Now, 
if $s=j_l$, for some $l=1,\ldots,k$, then $v=v_l$. On the other hand, if $s\notin \{j_1,\ldots,j_k\}$, then 
by assumption, there will exist some $j_l$ such that $\ell_N(j_l,s)=1$ and $h_N(j_l,s)=0$, that is, 
such that $v_l$ is a strict ancestor of $s$. This implies that $v=v_l$. Indeed, if $v\neq v_l$, then either
$v_l$ is an intermediate node in the path $v\pathgr s$, and in particular a tree descendant of $v$, which is forbidden by the time consistency because $v$ and $v_l$ have the hybrid child $i$ in common,
or $v$ is a proper descendant of $v_l$ through a path where $v_l$ and all the intermediate nodes are hybrid (if some of these nodes were of tree type, the temporal representation of $v$ would be greater than that of $v_l$, contradicting again the time consistency), in which case the child of $v_l$ in this path would be a hybrid child of $v_l$ that is a strict descendant of it (because it is intermediate in the path $v_l\pathgr v\pathgr s$ and $s$ is a strict descendant of $v_l$), which is impossible by Lemma  \ref{lem:hyb-nostr}.
\end{itemize}
This finishes the proof of (1).  

As far as (2) goes, in
$N_{H(i;j_1,\ldots,j_k)}$ we simply remove the hybrid leaf $i$
without removing anything else, and therefore no splitted LCSA-path length of a pair of remaining leaves   is affected.
\end{proof}

\begin{theorem}
\label{prop:D-arb-net}
Let $N$ and $N'$ be two TCTC-networks on the same set $S$ of taxa.
Then, $L^s(N)=L^s(N')$ if, and only if, $N\cong N'$.
\end{theorem}

\begin{proof}
The `if' implication is obvious.  We prove the `only if' implication
by double induction on the number $n$ of elements of $S$ and the
number $m$ of internal nodes of $N$.

As in Proposition \ref{prop:D-bin-net-H}, the cases $n=1$ and $n=2$
are straightforward, because both $\TCTC_1$ and $\TCTC_2$ consist of a single network.

On the other hand, the case when $m=1$, for every $n$, is also
straightforward: assuming $S=\{1,\ldots,n\}$, the network $N$ is in this case the phylogenetic tree with Newick
string \texttt{(1,2,\ldots,n);}, consisting only of the root and the
leaves, and in particular $L^s_N(i,j)=(1,1,0)$ for every $1\leq
i<j\leq n$.  If $L^s(N)=L^s(N')$, we have that $L^s_{N'}(i,j)=(1,1,0)$
for every $1\leq i<j\leq n$, and therefore all leaves in $N'$ are tree nodes and
sibling of each other by Lemma \ref{lem:hyb-nostr}.  Since the
root of a hybridization network cannot be elementary, this says that
$N'$ is also a phylogenetic tree with Newick
string \texttt{(1,2,\ldots,n);} and hence it is
isomorphic to $N$.

Let now $N$ and $N'$ two TCTC-networks with $n\geq 3$ leaves such that
$L^s(N)=L^s(N')$  and $N$ has $m\geq 2$   internal nodes.  Assume as induction hypothesis that the thesis in the theorem is
true for pairs of TCTC-networks $N_1,N_1'$ with $n-1$ leaves or with $n$ leaves and such
that $N_1$ has $m-1$ internal nodes.

By Proposition \ref{thm-reduction-possible}, a reduction $U(i)$, $T(i;j)$ or $H(i;j_1,\ldots, j_k)$ can be applied to $N$. Since the application conditions for such a reduction depend only on the splitted LCSA-path lengths vectors by Lemmas \ref{lem:red-U}.(1), \ref{lem:red-T}.(1) and
\ref{lem:red-H}.(1), and $L^s(N)=L^s(N')$, we conclude that we can apply the same reduction to $N'$.

Now, we apply the same reduction to $N$ and $N'$ to obtain new TCTC-networks $N_1$ and $N_1'$, respectively. If the reduction was of the form $U(i)$, $N_1$ and $N_1'$ have $n$ leaves and $N_1$ has $m-1$ internal nodes; if the reduction was of the forms $T(i;j)$ or $H(i;j_1,\ldots,j_k)$, $N_1$ and $N_1'$ have $n-1$ leaves. In all cases, $L^s(N_1)=L^s(N_1')$ by 
Lemmas \ref{lem:red-U}.(2), \ref{lem:red-T}.(2) and
\ref{lem:red-H}.(2), and therefore, by the induction hypothesis, $N_1\cong N_1'$.

Finally, by Lemma \ref{lem:iso-red}, $N$ and $N'$ are obtained from $N_1$ and $N_1'$ by applying the same expansion
$\textrm{U}^{-1}$, $\textrm{T}^{-1}$, or $\textrm{H}^{-1}$, and they are isomorphic. 
\end{proof}

The vectors of splitted LCSA-path lengths vectors do not separate hybridization networks much more general than the TCTC, as we following examples show.

\begin{remark}
The vectors of splitted distances do not separate arbitrary (that, is, possibly time inconsistent) tree-child phylogenetic networks. Indeed, the non-isomorphic tree-child binary phylogenetic networks $N$ and $N'$ depicted in Fig.~\ref{fig:notc} have the same $L^s$ vectors:
$$
L^s(N)=L^s(N')=\big((2,1,1),(4,1,1),(3,1,1)\big)
.
$$
\end{remark}

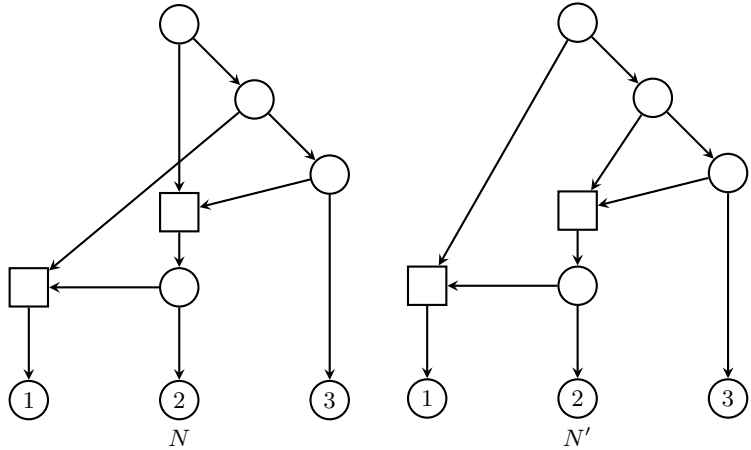
\begin{figure}[htb]
\begin{center}
%
              \begin{tikzpicture}[thick,>=stealth,scale=0.5]
              \draw(0,0) node[tre] (1) {};  \etq 1
              \draw(4,0) node[tre] (2) {};  \etq 2
              \draw(8,0) node[tre] (3) {};  \etq 3  
              \draw(0,3) node[hyb] (A) {};  
              \draw(4,3) node[tre] (u) {}; 
             \draw(4,5) node[hyb] (B) {};  
             \draw(8,6) node[tre] (t) {};  
             \draw(6,8) node[tre] (s) {};  
             \draw(4,10) node[tre] (r) {}; 
           \draw[->] (A)--(1);
            \draw[->] (u)--(A);
             \draw[->] (u)--(2);
            \draw[->] (B)--(u);
            \draw[->] (s)--(t);
            \draw[->] (t)--(B);
            \draw[->] (s)--(A);
             \draw[->] (r)--(s);
            \draw[->] (r)--(B);
            \draw[->] (t)--(3);
 \draw(4,-1) node  {$N$};
            \end{tikzpicture}
            \qquad            \begin{tikzpicture}[thick,>=stealth,scale=0.5]
              \draw(0,0) node[tre] (1) {};  \etq 1
              \draw(4,0) node[tre] (2) {};  \etq 2
              \draw(8,0) node[tre] (3) {};  \etq 3  
              \draw(0,3) node[hyb] (A) {};  
              \draw(4,3) node[tre] (u) {}; 
             \draw(4,5) node[hyb] (B) {};  
             \draw(8,6) node[tre] (t) {};  
             \draw(6,8) node[tre] (s) {};  
             \draw(4,10) node[tre] (r) {}; 
           \draw[->] (A)--(1);
            \draw[->] (u)--(A);
             \draw[->] (u)--(2);
            \draw[->] (B)--(u);
            \draw[->] (s)--(t);
            \draw[->] (t)--(B);
            \draw[->] (s)--(B);
             \draw[->] (r)--(s);
            \draw[->] (r)--(A);
            \draw[->] (t)--(3);
 \draw(4,-1) node  {$N'$};
            \end{tikzpicture}

\end{center}
\caption{\label{fig:notc}
These two tree-child binary phylogenetic networks have the same splitted LCSA-path lengths vectors.}
\end{figure}

\begin{remark}
The splitted LCSA-path lengths vectors do not separate {tree-sibling}   time consistent  phylogenetic networks, either.
Consider for instance the tree-sibling time consistent fully resolved phylogenetic networks $N$ and $N'$ depicted in Figure \ref{fig:treesib}. A simple computation shows that they have the same $L^s$ vectors, but they are not isomorphic. 
\end{remark}

\begin{figure}[htb]
\begin{center}
            \begin{tikzpicture}[thick,>=stealth,scale=0.5]
              \draw(0,0) node[tre] (1) {};  \etq 1
               \draw(2,0) node[tre] (2) {};  \etq 2
               \draw(6,0) node[tre] (3) {};  \etq 3
               \draw(8,0) node[tre] (4) {};  \etq 4
               \draw(10,0) node[tre] (5) {};  \etq 5
               \draw(12,0) node[tre] (6) {};  \etq 6
               \draw(16,0) node[tre] (7) {};  \etq 7
               \draw(18,0) node[tre] (8) {};  \etq 8
               \draw(0,3) node[tre] (a) {};  
               \draw(4,3) node[tre] (b) {};  
                  \draw(8,3) node[tre] (c) {};  
               \draw(10,3) node[tre] (d) {};  
               \draw(14,3) node[tre] (e) {};  
                  \draw(18,3) node[tre] (f) {};  
                  \draw(3,6) node[tre] (g) {};  
                   \draw(7,6) node[tre] (h) {};  
                   \draw(5,7) node[tre] (i) {};  
                  \draw(13,7) node[tre] (j) {};  
                  \draw(9,9) node[tre] (r) {};  
                  \draw(2,2) node[hyb] (A) {};  
                  \draw(6,2) node[hyb] (B) {};  
                  \draw(12,2) node[hyb] (C) {};  
                  \draw(16,2) node[hyb] (D) {};  
               \draw[->] (a)--(1);
            \draw[->] (a)--(A);
             \draw[->] (A)--(2);
            \draw[->] (b)--(A);
            \draw[->] (b)--(B);
              \draw[->] (B)--(3);
           \draw[->] (c)--(B);
            \draw[->] (c)--(4);
                \draw[->] (d)--(5);
            \draw[->] (d)--(C);
             \draw[->] (C)--(6);
            \draw[->] (e)--(C);
            \draw[->] (e)--(D);
              \draw[->] (D)--(7);
           \draw[->] (f)--(D);
            \draw[->] (f)--(8);
            \draw[->] (g)--(a);
            \draw[->] (g)--(f);
            \draw[->] (h)--(c);
             \draw[->] (h)--(d);
            \draw[->] (i)--(g);
            \draw[->] (i)--(h);
            \draw[->] (j)--(b);
            \draw[->] (j)--(e);
            \draw[->] (r)--(i);
            \draw[->] (r)--(j);
   \draw(9,-1) node  {$N$};
            \end{tikzpicture}
  \vspace*{3ex}
  
             \begin{tikzpicture}[thick,>=stealth,scale=0.5]
              \draw(0,0) node[tre] (1) {};  \etq 1
               \draw(2,0) node[tre] (2) {};  \etq 2
               \draw(6,0) node[tre] (7) {};  \etq 7
               \draw(8,0) node[tre] (8) {};  \etq 8
               \draw(10,0) node[tre] (5) {};  \etq 5
               \draw(12,0) node[tre] (6) {};  \etq 6
               \draw(16,0) node[tre] (3) {};  \etq 3
               \draw(18,0) node[tre] (4) {};  \etq 4
               \draw(0,3) node[tre] (a) {};  
               \draw(4,3) node[tre] (b) {};  
                  \draw(8,3) node[tre] (c) {};  
               \draw(10,3) node[tre] (d) {};  
               \draw(14,3) node[tre] (e) {};  
                  \draw(18,3) node[tre] (f) {};  
                  \draw(3,6) node[tre] (g) {};  
                   \draw(7,6) node[tre] (h) {};  
                   \draw(5,7) node[tre] (i) {};  
                  \draw(13,7) node[tre] (j) {};  
                  \draw(9,9) node[tre] (r) {};  
                  \draw(2,2) node[hyb] (A) {};  
                  \draw(6,2) node[hyb] (B) {};  
                  \draw(12,2) node[hyb] (C) {};  
                  \draw(16,2) node[hyb] (D) {};  
               \draw[->] (a)--(1);
            \draw[->] (a)--(A);
             \draw[->] (A)--(2);
            \draw[->] (b)--(A);
            \draw[->] (b)--(B);
              \draw[->] (B)--(7);
           \draw[->] (c)--(B);
            \draw[->] (c)--(8);
                \draw[->] (d)--(5);
            \draw[->] (d)--(C);
             \draw[->] (C)--(6);
            \draw[->] (e)--(C);
            \draw[->] (e)--(D);
              \draw[->] (D)--(3);
           \draw[->] (f)--(D);
            \draw[->] (f)--(4);
            \draw[->] (g)--(a);
            \draw[->] (g)--(c);
            \draw[->] (h)--(d);
             \draw[->] (h)--(f);
            \draw[->] (i)--(g);
            \draw[->] (i)--(h);
            \draw[->] (j)--(b);
            \draw[->] (j)--(e);
            \draw[->] (r)--(i);
            \draw[->] (r)--(j);
   \draw(9,-1) node  {$N'$};
            \end{tikzpicture}

\end{center}
\caption{\label{fig:treesib}
These two tree-sibling time consistent binary phylogenetic networks have the same splitted LCSA-path lengths vectors.}
\end{figure}
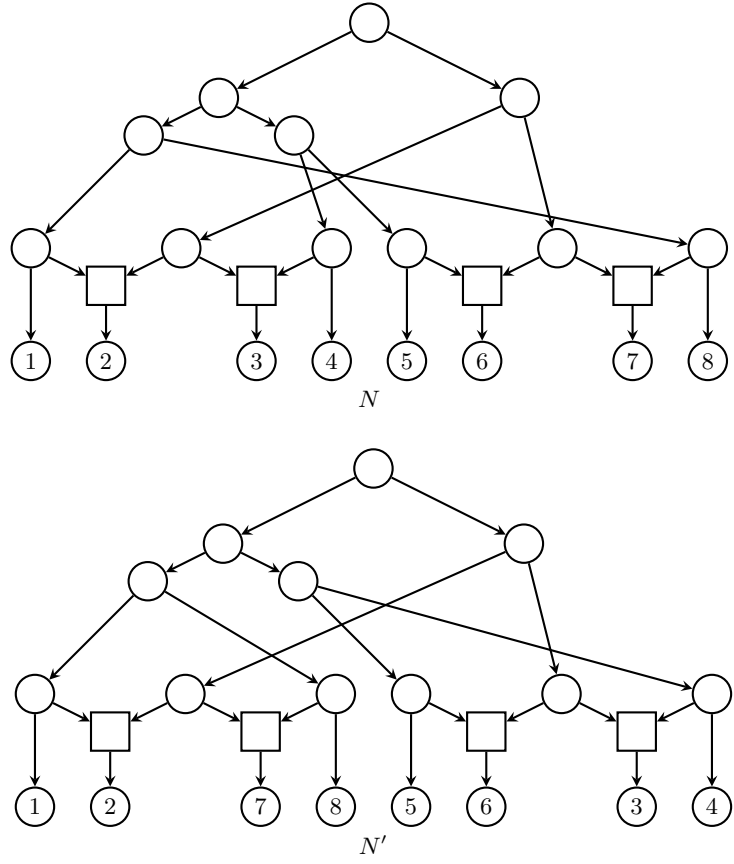

As in the fully resolved case, the injectivity of the mapping
$$
L^s:\TCTC_n\to  \RR^{3n(n-1)/2}
$$
makes it possible to induce metrics on $\TCTC_n$ from metrics on $\RR^{3n(n-1)/2}$.
The proof of the following result is similar to that of Proposition \ref{prop:metric}.

\begin{proposition}
For every $n\geq 1$,  let $D$ be any metric on $\RR^{3n(n-1)/2}$. The mapping
$d^s: \TCTC_n\times \TCTC_n  \to  \RR$ defined by
$d (N_1,N_2) = D(L^s(N_1),L^s(N_2))$
 satisfies the axioms of metrics up to isomorphisms. \qed
  \end{proposition}

For instance, using as $D$ the Manhattan distance or the Euclidean distance, we obtain, respectively, the metrics on $\TCTC_n$
$$
\begin{array}{rl}
\displaystyle d^s_1(N_1,N_2)& \displaystyle
=\sum_{1\leq i<  j\leq n}\big(|\ell_{N_1}(i,j)-\ell_{N_2}(i,j)|+|\ell_{N_1}(j,i)-\ell_{N_2}(j,i)|\\[-1ex]
&  \displaystyle \qquad\qquad\qquad\qquad \qquad\qquad \qquad +  |h_{N_1}(i,j)-h_{N_2}(i,j)|\big)\\[1ex]
& \displaystyle =\sum_{1\leq i\neq  j\leq n}\big(|\ell_{N_1}(i,j)-\ell_{N_2}(i,j)|
+\frac{1}{2} |h_{N_1}(i,j)-h_{N_2}(i,j)|\big)\\[3ex]
\displaystyle  d^s_2(N_1,N_2)& \displaystyle
=\Big(\sum_{1\leq i<  j\leq n}\big((\ell_{N_1}(i,j)-\ell_{N_2}(i,j))^2+(\ell_{N_1}(j,i)-\ell_{N_2}(j,i))^2\\[-1ex]
& \displaystyle\qquad\qquad\qquad\qquad\qquad\qquad \qquad  +(h_{N_1}(i,j)-h_{N_2}(i,j))^2\big)\Big)^{\frac 12}\\[1ex]
& \displaystyle =\sqrt{\sum_{1\leq i \neq  j\leq n}\big((\ell_{N_1}(i,j)-\ell_{N_2}(i,j))^2
+\frac{1}{2}(h_{N_1}(i,j)-h_{N_2}(i,j))^2\big)}
\end{array}
$$
These metrics generalize to TCTC-networks the splitted nodal metrics for arbitrary phylogenetic trees defined in \cite{cardona.ea:08a}. and the nodal metric for TCTC \emph{phylogenetic} networks defined in \cite{comparison2}.

\section{Conclusions}

A classical result of Smolenskii \cite{smolenskii:63} establishes that
the vectors of distances between pairs of leaves separate unrooted
phylogenetic trees on a given set of taxa.  This result generalizes
easily to fully resolved rooted phylogenetic trees
\cite{cardona.ea:08a}, and it lies at the basis of the classical
definitions of \emph{nodal distances} for unrooted as well as for
fully resolved rooted phylogenetic trees based on the comparison of
these vectors
\cite{bluis.ea:2003,farris:sz69,farris:sz73,phipps:sz71,steelpenny:sb93,willcliff:taxon71}.
But these vectors do not separate arbitrary rooted phylogenetic trees,
and therefore they cannot be used to compare the latter in a sound
way.  This problem was overcome in \cite{cardona.ea:08a} by
introducing the   \emph{splitted  path lengths}  matrices 
 and 
 showing that they separate arbitrary rooted
phylogenetic trees on a given set of taxa.  It is possible then to
define \emph{splitted nodal} metrics for arbitrary rooted phylogenetic
trees by comparing these matrices.

In this paper we have generalized these results to the class $\TCTC_n$ of
tree-child time consistent hybridization networks (TCTC-networks) with $n$ leaves. For every pair $i,j$ of leaves in a TCTC-network $N$, we have defined the \emph{LCSA-path length} $L_N(i,j)$ 
and the \emph{splitted LCSA-path length} $L_N^s(i,j)$
between $i$ and $j$
and we have proved that the vectors $L(N)=(L_N(i,j))_{1\leq
i<j\leq n}$ separate fully resolved networks in $\TCTC_n$ and
the vectors
$L^s(N)=(L_N^s(i,j))_{1\leq i<j\leq
n}$  separate arbitrary  TCTC-networks.

The vectors $L(N)$ and $L^s(N)$ can be computed in low polynomial time
by means of simple algorithms that do not require the use of
sophisticated data structures.  Indeed, let $n$ be the number of
leaves and $m$ the number of internal nodes in $N$.  As we explained
in \cite[\S V.D]{comparison1}, for each internal node $v$ and for each
leaf $i$, it can be decided whether $v$ is a strict or a non-strict
ancestor of $i$, or not an ancestor of it at all, by computing by
breadth-first search the shortest paths from the root to each leaf
before and after removing each of the $m$ nodes in turn, because a
non-strict descendant of a node will still be reachable from the root
after removing that node, while a strict descendant will not.  All
this information can be computed in $O(m(n+m))$ time, and once it has
been computed the least common semi-strict ancestor of two leaves can
be computed in $O(m)$ time by selecting the node of least height among
those which are ancestors of the two leaves and strict ancestors of at
least one of them.  This allows the computation of $L(N)$ and $L^s(N)$
in $O(m^2+n^2m)$ time.

These vectors $L(N)$ and $L^s(N)$ can be used then to define metrics
for fully resolved and arbitrary TCTC-networks, respectively, from
metrics for real-valued vectors.  The  metrics obtained in this way can be understood as
generalizations to $\TCTC_n$ of the (non-splitted or splitted) nodal metrics for phylogenetic
trees and they can be computed in low polynomial time if the metric
used to compare the vectors can be done so:
this is the case, for instance, when this metric is the Manhattan or
the Euclidean metric (in the last case, computing the square root with
$O(10^{m+n})$ significant digits \cite{batra:2008}, which should
be more than enough).

It remains to study the main properties of the  metrics   defined in this way, like for instance
their diameter or the distribution of their values.  It is important
to recall here that these are open problems even for the classical
nodal distances for fully resolved rooted phylogenetic trees.

\section*{Acknowledgment}

The research reported in this paper has been partially supported by
the Spanish 
DGI projects MTM2006-07773 COMGRIO and MTM2006-15038-C02-01.


\begin{thebibliography}{10}
\expandafter\ifx\csname url\endcsname\relax
  \def\url#1{\texttt{#1}}\fi
\expandafter\ifx\csname urlprefix\endcsname\relax\def\urlprefix{URL }\fi

\bibitem{baroni.ea:sb06}
M.~Baroni, C.~Semple, M.~Steel, Hybrids in real time, Syst.\ Biol. 55 (2006)
  46--56.

\bibitem{batra:2008}
P.~Batra, Newton's method and the computational complexity of the fundamental
  theorem of algebra, Electron. Notes Theor. Comput. Sci. 202 (2008) 201--218.

\bibitem{bluis.ea:2003}
J.~Bluis, D.-G. Shin, Nodal distance algorithm: Calculating a phylogenetic tree
  comparison metric, in: Proc.\ 3rd IEEE Symp.\ BioInformatics and
  BioEngineering, 2003.

\bibitem{cardona.ea:sbTSTC:2008}
G.~Cardona, M.~Llabr\'es, F.~Rossell\'o, G.~Valiente, A distance metric for a
  class of tree-sibling phylogenetic networks, Bioinformatics 24~(13) (2008)
  1481--1488.

\bibitem{comparison1}
G.~Cardona, M.~Llabr\'es, F.~Rossell\'o, G.~Valiente, Metrics for phylogenetic
  networks {I}: Generalizations of the {R}obinson-{F}oulds metric, submitted
  (2008).

\bibitem{comparison2}
G.~Cardona, M.~Llabr\'es, F.~Rossell\'o, G.~Valiente, Metrics for phylogenetic
  networks {II}: Nodal and triplets metrics, submitted (2008).

\bibitem{cardona.ea:08a}
G.~Cardona, M.~Llabr\'es, F.~Rossell\'o, G.~Valiente, Nodal metrics for rooted
  phylogenetic trees, submitted, available at \texttt{arxiv.org/abs/0806.2035}
  (2008).

\bibitem{cardona.ea:07b}
G.~Cardona, F.~Rossell\'o, G.~Valiente, Comparison of tree-child phylogenetic
  networks, IEEE T. Comput.\ Biol. preprint, 30 June 2008 ,
  doi:\texttt{10.1109/TCBB.2007.70270}.

\bibitem{cardona.ea:07a}
G.~Cardona, F.~Rossell\'o, G.~Valiente, Tripartitions do not always
  discriminate phylogenetic networks, Math. Biosci. 211~(2) (2008) 356--370.

\bibitem{doolittle:99}
W.~F. Doolittle, Phylogenetic classification and the universal tree, Science
  284~(5423) (1999) 2124--2128.

\bibitem{farris:sz69}
J.~S. Farris, A successive approximations approach to character weighting,
  Syst. Zool. 18 (1969) 374--385.

\bibitem{farris:sz73}
J.~S. Farris, On comparing the shapes of taxonomic trees, Syst. Zool. 22 (1973)
  50--54.

\bibitem{gusfield.ea:fine.structure:2004}
D.~Gusfield, S.~Eddhu, C.~Langley, The fine structure of galls in phylogenetic
  networks, INFORMS J.\ Comput, 16~(4) (2004) 459--469.

\bibitem{gusfield.ea:galled.trees:2004}
D.~Gusfield, S.~Eddhu, C.~Langley, Optimal, efficient reconstruction of
  phylogenetic networks with constrained recombination, J. Bioinformatics
  Comput.\ Biol. 2~(1) (2004) 173--213.

\bibitem{coalescent}
J.~Hein, M.~H. Schierup, C.~Wiuf, Gene Genealogies, Variation and Evolution: A
  Primer in Coalescent Theory, Oxford University Press, 2005.

\bibitem{huson.ea:06}
D.~H. Huson, D.~Bryant, {Application of Phylogenetic Networks in Evolutionary
  Studies}, Mol.\ Biol.\ Evol. 23~(2) (2006) 254--267.

\bibitem{huson:RECOMB07}
D.~H. Huson, T.~H. Kl{\"o}pper, Beyond galled trees - decomposition and
  computation of galled networks, in: Proceedings RECOMB 2007, vol. 4453 of
  Lecture Notes in Computer Science, Springer-Verlag, 2007.

\bibitem{moret.ea:2004}
B.~M.~E. Moret, L.~Nakhleh, T.~Warnow, C.~R. Linder, A.~Tholse, A.~Padolina,
  J.~Sun, R.~Timme, Phylogenetic networks: Modeling, reconstructibility, and
  accuracy, IEEE T. Comput.\ Biol. 1~(1) (2004) 13--23.

\bibitem{nakhleh.ea:2003}
L.~Nakhleh, J.~Sun, T.~Warnow, C.~R. Linder, B.~M.~E. Moret, A.~Tholse, Towards
  the development of computational tools for evaluating phylogenetic network
  reconstruction methods, in: Proc.\ 8th Pacific Symp.\ Biocomputing, 2003.

\bibitem{nakhleh.ea:03psb}
L.~Nakhleh, J.~Sun, T.~Warnow, C.~R. Linder, B.~M.~E. Moret, A.~Tholse, Towards
  the development of computational tools for evaluating phylogenetic network
  reconstruction methods, in: Proc.\ 8th Pacific Symp.\ Biocomputing, 2003.

\bibitem{nakhleh.ea:2005}
L.~Nakhleh, T.~Warnow, C.~R. Linder, K.~S. John, Reconstructing reticulate
  evolution in species: Theory and practice, J. Comput.\ Biol. 12~(6) (2005)
  796--811.

\bibitem{phipps:sz71}
J.~B. Phipps, Dendrogram topology, Syst. Zool. 20 (1971) 306--308.

\bibitem{semple:07}
C.~Semple, Hybridization networks, in: O.~Gascuel, M.~Steel (eds.),
  Reconstructing evolution: New mathematical and computational advances, Oxford
  University Press, 2008, pp. 277--314.

\bibitem{smolenskii:63}
Y.~A. Smolenskii, A method for the linear recording of graphs, USSR
  Computational Mathematics and Mathematical Physics 2 (1963) 396--397.

\bibitem{song.ea:2005}
Y.~S. Song, J.~Hein, Constructing minimal ancestral recombination graphs, J.
  Comput.\ Biol. 12~(2) (2005) 147--169.

\bibitem{steelpenny:sb93}
M.~A. Steel, D.~Penny, Distributions of tree comparison metrics---some new
  results, Syst. Biol. 42~(2) (1993) 126--141.

\bibitem{valiente:08}
G.~Valiente, Phylogenetic networks, course at the Int. Summer School on
  Bioinformatics and Computational Biology Lipari  (June 14--21, 2008).

\bibitem{wang.ea:2001}
L.~Wang, K.~Zhang, L.~Zhang, Perfect phylogenetic networks with recombination,
  J. Comput.\ Biol. 8~(1) (2001) 69--78.

\bibitem{willcliff:taxon71}
W.~T. Williams, H.~T. Clifford, On the comparison of two classifications of the
  same set of elements, Taxon 20~(4) (1971) 519--522.

\bibitem{willson:07}
S.~J. Willson, Restrictions on meaningful phylogenetic networks, contributed
  talk at the {EMBO} Workshop on Current Challenges and Problems in
  Phylogenetics (Isaac Newton Institute for Mathematical Sciences, Cambridge,
  UK, 3--7 September 2007).

\bibitem{willson:07b}
S.~J. Willson, Reconstruction of certain phylogenetic networks from the genomes
  at their leaves, J. Theor.\ Biol. 252 (2008) 338--349.

\bibitem{woolley.ea:2008}
S.~M. Woolley, D.~Posada, K.~A. Crandall, A comparison of phylogenetic network
  methods using computer simulation, Plos ONE 3~(4) (2008) e1913.

\end{thebibliography}
\end{document}